\documentclass[11pt, a4paper]{article}

\usepackage{microtype}
\usepackage{fullpage}
\usepackage{authblk}
\usepackage{amssymb}
\usepackage[english]{babel}
\usepackage[utf8]{inputenc}
\usepackage{graphicx}
\usepackage[colorinlistoftodos]{todonotes}
\usepackage{mathtools}
\usepackage{enumerate}
\usepackage{lineno}
\usepackage{amsmath,amsfonts,amsthm}
\usepackage{thmtools, thm-restate}
\usepackage{cancel}
\usepackage{xspace}
\usepackage{ifthen}
\usepackage{times}

\usepackage{color}
\definecolor{darkgreen}{rgb}{0,0.5,0}
\definecolor{darkblue}{rgb}{0,0,0.8}
\definecolor{darkred}{rgb}{0.8,0,0}

\newcommand{\many}{^\#\!\!}

\usepackage{hyperref}
\hypersetup{
   unicode=false,          
   colorlinks=true,        
   linkcolor=darkred,          
   citecolor=darkblue,        
   filecolor=magenta,      
   urlcolor=black           
}

\newtheorem{definition}{Definition}[section]

\newtheorem{theorem}[definition]{Theorem}

\newtheorem{proposition}[definition]{Proposition}

\newcommand{\ignore}[1]{}

\newcommand{\poly}{\mathrm{poly}}
\newcommand{\eps}{\varepsilon}
\newcommand{\set}[1]{\ensuremath \mathcal{#1}}

\usepackage{fancyvrb}
\usepackage[scaled=.85]{beramono}
\usepackage[T1]{fontenc}

\newcommand{\key}[1]{\protect\textbf{#1}}
\newcommand{\fun}[1]{{\sf #1}}

\newcommand{\THREAD}[2]{\key{thread} $\fun{#1}$  $\protect\key{uses}$   $#2$:}
\newcommand{\VAR}[1]{\key{var} $#1$:}
\newcommand{\LOCAL}[1]{\key{var} $#1$}
\newcommand{\READONLY}[1]{\key{reads } #1}

\newcommand{\GETS}[2]{$\ensuremath{#1 := #2}$}
\newcommand{\INIT}[2]{$\ensuremath{#1 \gets {\textit{#2}}}$}

\newcommand{\PROTOCOL}[1]{\key{def protocol} $\fun{#1}$}

\newcommand{\REPEAT}[1][]{\key{repeat}\ifthenelse{\equal{#1}{}}{\key{}}{ $\geq #1$ \key{times}}:}
\newcommand{\ASYNC}[1][]{\key{execute}\ifthenelse{\equal{#1}{}}{\key{}}{\key{ for $\geq #1$ rounds}} \key{ruleset}:}

\newcommand{\IF}{\key{if}:}
\newcommand{\IFNONEMPTY}[1]{\key{if exists} {\ensuremath{(#1)}}:}
\newcommand{\ELSE}{\key{else}:}
\newcommand{\INPUT}[1]{#1\ \key{as input}}
\newcommand{\OUTPUT}[1]{#1\ \key{as output}}
\newcommand{\RULE}[4]{$\triangleright$ $(#1)$ + $(#2)$ $\to$ $(#3)$ + $(#4)$}

\newcommand{\LeaderElection}{LeaderElection}
\newcommand{\Majority}{Majority}
\newcommand{\LeaderElectionExact}{LeaderElectionExact}
\newcommand{\MajorityExact}{MajorityExact}

\newcommand{\SemilinearExact}{SemilinearPredicateExact}

\newcommand{\Main}{Main}
\newcommand{\FilteredCoin}{FilteredCoin}
\newcommand{\ReduceSets}{ReduceSets}
\newcommand{\Trim}{Trim}

\newcommand{\ID}{\ensuremath{^\#\!\!\!\!}}

\newcommand{\on}{\ensuremath{\textit{on}}}
\newcommand{\off}{\ensuremath{\textit{off}}}

\DefineVerbatimEnvironment%
{protocol}{Verbatim}%
{tabsize=4,numbers=left,fontfamily=fvm,numbersep=2mm,frame=lines,framerule=1pt,%
codes={\catcode`$=3\catcode`_=8\catcode`^=7\catcode`'=12},commandchars=\\\{\}}

\title{\bf Population Protocols Are Fast}

\author[1]{Adrian Kosowski}
\affil[1]{Inria Paris
, France}
\author[2]{Przemys\l{}aw~Uzna\'nski}
\affil[2]{
ETH Zürich, Switzerland}
\date{}
\begin{document}
\maketitle

\thispagestyle{empty}

\begin{abstract}
A population protocol describes a set of state change rules for a population of $n$ indistinguishable finite-state agents (automata), undergoing random pairwise interactions. Within this very basic framework, it is possible to resolve a number of fundamental tasks in distributed computing, including: leader election, aggregate and threshold functions on the population, such as majority computation, and plurality consensus. For the first time, we show that solutions to all of these problems can be obtained \emph{quickly} using finite-state protocols. For any input, the designed finite-state protocols converge under a fair random scheduler to an output which is correct with high probability in expected $O(\poly \log n)$ parallel time. In the same setting, we also show protocols which always reach a valid solution, in expected parallel time $O(n^\eps)$, where the number of states of the interacting automata depends only on the choice of $\eps>0$. The stated time bounds hold for \emph{any} semi-linear predicate computable in the population protocol framework.

The key ingredient of our result is the decentralized design of a hierarchy of phase-clocks, which tick at different rates, with the rates of adjacent clocks separated by a factor of $\Theta(\log n)$. The construction of this clock hierarchy relies on a new protocol composition technique, combined with an adapted analysis of a self-organizing process of oscillatory dynamics. This clock hierarchy is used to provide nested synchronization primitives, which allow us to view the population in a global manner and design protocols using a high-level imperative programming language with a (limited) capacity for loops and branching instructions. 

\end{abstract}

\clearpage
\setcounter{page}{1}

\section{Introduction}


The model of population protocols was introduced by Angluin et al.~\cite{DBLP:journals/dc/AngluinADFP06}. A population consists of a large collection of indistinguishable agents that rely on very limited communication and computation power. Agents interact in a pairwise, sequential manner, governed by a scheduler. Rules of the protocol determine how agents update their states, and the update depends only on their states at the moment of interaction. The population protocol framework is particularly well-adapted to applications in interacting particle systems, which includes modeling behavior of biological agents and the programming of chemical systems. Specifically, the population protocol framework is equivalent to the formalism of fixed-volume Chemical Reaction Networks  (CRN's,~\cite{DBLP:journals/dc/ChenCDS17}), and may be used directly for programming in frameworks such as DNS strand computing.

From a computational perspective, we consider the standard input/output population protocol framework in which input is encoded in the form of the starting states of the population of agents, and output is decoded from the population states after a certain time.\footnote{Most simply, one can consider the state of an agent as a $2$-tuple, in which the first entry encodes its ``working state'' and the second its current output.} A basic set of tasks is that of computing \emph{predicates} on the number of agents initially occupying particular input states, so that the outputs of all agents eventually agree on the value of the predicate.\footnote{We use predicates only as a particularly representative example of tasks which can be resolved, noting that some natural tasks of distributed computing, e.g., leader election, are not representable as computational predicates, since not all agents are expected to converge to the same output value.}

The most prevalent scheduler model, also considered in this paper, uses a probabilistic scheduler, where pairs of interacting agents are selected uniformly at random. In this model, one is interested in the \emph{convergence time}, being the number of interactions $t$ until all agents agree on the same output values, which subsequently never change for any agent. Usually, the measure  $t/n$ called \emph{parallel time} is considered, which corresponds to the notion of the number of \emph{rounds} of parallel interactions.  

In this paper we look at the interplay between convergence time and the number of states of the agent's machine. The original formulation assumes that the agents are automata with a constant number of states~\cite{DBLP:journals/dc/AngluinADFP06}. Since then, this assumption has been frequently relaxed in the literature, making the number of states slightly dependent on the population size $n$.
Such a relaxation has allowed for significant progress in the field, and specifically for designs of rapidly converging protocols for basic tasks. These results include majority computation in $O(\log^2 n)$ parallel time by the protocol of Alistarh et al.~\cite{DBLP:conf/soda/AlistarhAG18} using $O(\log n)$ states, and leader election in $O(\log^2 n)$ parallel time with the protocol of Gąsieniec and Stachowiak~\cite{DBLP:conf/soda/GasieniecS18} using $O(\log \log n)$ states, under a fair random scheduler. Both these tasks can also be solved using finite-state agents, but the best known solutions to date required super-linear parallel time.

We remark that in all relevant applications of population protocols related to modeling and design of complex systems, agents are naturally described as finite-state machines. In this light, results in the relaxed model with a super-constant number of states have  a more limited explanatory power for real-world phenomena, and they perhaps more pertinently seen as lower-bounding the population size $n$ for which a $k$-state protocol still operates correctly, for fixed $k$. This is sometimes sufficient to explain processes at reasonably large scales, but not in the macroscopic limit $n \to +\infty$. The model relaxation, up to a state space of polylogarithmic size, also does not seem particularly useful in terms of enlarging the set of tasks which can be resolved in the model.\footnote{Cf.~e.g.~\cite{DBLP:journals/tcs/ChatzigiannakisMNPS11} for an equivalence of the languages of predicates computable using finite state protocols and protocols with less than a polylogarithmic number of states, under some important additional assumptions.}

This paper shows that \emph{it is possible to achieve fast computation in the finite-state model} of population protocols. Previously, results of this type were only known in the model variant with a unique pre-distinguished leader agent, which directly allows for a form of centralized control~\cite{DBLP:journals/dc/AngluinAE08a}. Here, we provide a completely leaderless approach to bypass the lack of synchronicity of the system, and consequently the lack of a common time measure for agents. This is achieved by creating a hierarchy of so-called \emph{phase-clocks}, running at carefully tuned rates, which synchronize agents sufficiently to allow a form of imperative centralized program to be deployed on the system. 

\subsection{Overview of Our Results}

We provide a general framework for programming in what we call \emph{a sequential code} for population protocols, and describe a way of compilation of such code to a set of rules ready for execution in the protocol environment. The language itself involves finite-depth loops of prescribed length (lasting a logarithmic number of rounds, w.h.p.) and branching instructions. It is meant to have sufficient expressive power on the one hand, and to bound the expected execution time of any finite-state protocol expressed in it, with the precise trade-off depending on the adopted variant of compilation. 

\paragraph{Results for w.h.p.\ correct protocols.} 

The basic compilation scheme provides a simple way of achieving protocols which give results which are correct w.h.p., using $O(1)$ states and converging in expected polylogarithmic parallel time. 

As a representative example, we express in the designed language protocols for two basic tasks: leader election and majority. \emph{Leader election} is a problem in which all agents initially hold the same input state, and the goal is for the population to reach a configuration in which exactly one agent holds a distinguished output, labeled as the leader. \emph{Majority} (in its generalized comparison version) is a problem in which some subset of agents holds initial state $A$, a disjoint subset of the population (of different cardinality) holds initial state $B$, and the population should converge to an identical boolean output for all agents, depending only on which of the initial sets is the larger one. 

As a direct consequence of the  formulation of our protocols within the framework (with no need for further analysis), we obtain the first constant-state protocols for leader election and majority, converging in polylogarithmic expected time w.h.p. to a w.h.p. correct result, under a fair scheduler (see Section~\ref{le} and Section~\ref{maj}). We remark that for majority, the high probability of correctness holds regardless of the gap between the sizes of the compared sets. 

We remark that in the framework, the precise convergence time of the protocols can be given as $O(\log^{c+1} n)$ rounds, where $c$ is the depth of nested ``repeat'' loops in the formulation of the protocol. In the adopted implementations it is thus $O(\log^2 n)$ rounds for leader election and $O(\log^3 n)$ rounds for majority.

Next, the obtained solution to leader election allows us to exploit the leader-driven programming framework of~\cite{DBLP:journals/dc/AngluinAE08a}, and to combine it with our framework (see Section~\ref{general-semi-linear}). We apply protocols generated with (compiled within) this framework as a blackbox, composed with the solution to leader election. This allows us to compute any semi-linear predicate (i.e., one which can be resolved using  population protocols with $O(1)$ states~\cite{DBLP:journals/dc/AngluinADFP06}) using a protocol converging in polylogarithmic time ($O(\log^5 n)$ rounds). Such predicates are much more general than the majority problem, including  threshold estimation of the size of a set in relation to the entire population, as well as computation of modulo remainders. The task of \emph{plurality consensus} (cf.~\cite{DBLP:conf/soda/BecchettiCNPS15,DBLP:conf/icalp/BerenbrinkFGK16,DBLP:conf/podc/GhaffariP16a}), in which the goal is to identify the largest of $l$ input sets, can also be expressed using semi-linear predicates. A solution to plurality consensus is in fact obtained with a straightforward adaptation of our protocol for majority, with the same convergence time.\footnote{We remark that the number of states in the solution to plurality consensus will depend on $l$, and after some optimization can be bounded as $O(l^2)$. We leave as open the question if this dependence is optimal in the studied setting.}

\paragraph{Results for always-correct protocols.} 

All the designed protocols can be made correct with certainty, converging in $O(n^\varepsilon)$ time, where $\eps>0$ is an arbitrarily chosen parameter, influencing the number of states of the agent. Alternatively, polylogarithmic running time can be preserved at the cost of using $O(\log \log n)$ states per agent.

\paragraph{Techniques, proof outline, and originality.}

The idea of programming protocols using clock hierarchies goes back to the leader-based framework of~\cite{DBLP:journals/dc/AngluinAE08a}. With respect to~\cite{DBLP:journals/dc/AngluinAE08a}, the main new element is the operation of the phase-clock hierarchy itself, which is different and more involved in terms of its (leaderless) design. It also provides stronger synchronization guarantees than those required in~\cite{DBLP:journals/dc/AngluinAE08a}, promising that all agents are executing the same line of code at the same time, w.h.p.

Our clock design hierarchy itself relies on several building blocks. We first revisit the analysis of a self-stabilizing 7-state oscillatory protocol of the authors~\cite{DBLP:journals/corr/DudekK17}, and extend this oscillating mechanism to obtain a modulo-3 phase clock, operating correctly in the presence of the correct concentration of agents in the clock's designated \emph{control state $X$} in the population. We then show how to extend the modulo-3 phase clock into a modulo-$m$ phase-clock, for $m$ being an arbitrarily large constant (cf. Section~\ref{sec:clocks}). When the clock is operating in the correct conditions and has stabilized to its intended behavior, this allows all agents to agree as to the modulo-$m$ phase indicated by the clock (up to a difference of at most $1$), w.h.p.

The construction of the hierarchy now relies on a mechanism to drive the clocks. The internals of the base modulo-$m$ clock progress following the asynchronous random scheduler of the system. The phase of the clock progresses every $\Theta(\log n)$ rounds w.h.p. when the control state $X$ is represented by the right number $\many X$ of agents in the population, contained in the range $\many X \in [1,n^{1-\eps}]$, where $\eps>0$ is an arbitrarily fixed design parameter of the clock (see Theorem~\ref{th:variantofclock}). For subsequent clocks in the hierarchy, we use the same control state $X$, but slow down the execution of rules of the $(i+1)$-st clock, using the $i$-th clock so that $\Theta(n)$ rules of the $(i+1)$-st clock are executed during a single period of the $i$-th clock. In this way, the period of the $i$-th clock in the hierarchy is $\Theta(\log n)$ rounds, for $i=1,2,3,\ldots$. The number of clocks placed in the hierarchy depends on the depth of nested loops in the executed program.

To the best of our knowledge, this is the first systematic  approach to decentralized clock composition in the literature.\footnote{For a simpler ad-hoc composition of two non-self-stabilizing clocks with different designs, see e.g.~\cite{DBLP:conf/soda/GasieniecS18,GSU18}.} We remark on two essential points in the construction.

First of all, compositions of clocks (or, for that matter, or most other population protocols) are notoriously hard to analyse rigorously. Here, we achieve this by using each clock to emulate a slower scheduler for the next clock in the population, keeping the clock protocol otherwise independent. We use to our advantage the fact that all of the composed protocols have a finite number of states, which allows us in some crucial parts to rely on the continuous approximation of the protocol, corresponding to the limit $n\to +\infty$ (mean-field approximation). For a $k$-state protocol with $k$ fixed, we then identify the state of the population with a point in the phase space $[0,1]^k$, with each coordinate corresponding to the proportion of the number of agents holding the given state in the population. The evolution of states can then be approximated using the continuous dynamics (set of ordinary differential equations) corresponding to the continuous limit of the protocol. The  scope of applicability of this type of approximation is by now well understood for population protocols (cf.~e.g.~\cite{DBLP:conf/icalp/CzyzowiczGKKSU15,DBLP:conf/mfcs/BournezFK12,DBLP:journals/corr/DudekK17}). Informally, over time scales of the length analyzed in this construction (parallel time polynomial in $n$), the behavior of the system 
in parts of the phase space, which are separated from its singularities and fixed points by a distance of at least $n^{-0.5+\Omega(1)}$, can be analyzed using standard concentration bounds (cf.~e.g.~\cite{DBLP:journals/corr/DudekK17}[Lemma~1 and subsequent applications of Azuma's inequality]). The clock protocols we use do not have to operate all the time in such an area of the state space, but once they start correct operation, they would need to cross such an area in order to display subsequent incorrect behavior, due to which such behavior becomes a low probability event. The analysis obtained through continuous approximation is robust and holds over different types of random schedulers. It generalizes naturally from an asynchronous scheduler to a parallel scheduler, which activates a random matching in the population in every step. This is crucial, since we do not know how to simulate a slowed version of the asynchronous scheduler to propagate successive clocks, but succeed in using clocks to emulate an (almost perfect) parallel matching scheduler in the population, working at a rate of one activation of the scheduler per period of the clock, thus giving the required slowdown factor of $\Theta(\log n)$ in the construction of the hierarchy  (see Section~\ref{hierarchia} for details).

The second crucial aspect concerns maintaining the correct concentration of agents in state $X$, which controls the clocks in the hierarchy. We start by remarking on the reason for our choice of the clock mechanism. It is relatively easy~\cite{DBLP:journals/dc/AngluinAE08a} to design phase clocks which run in a setting with a unique leader ($\many X=1$), however (as suggested before) the need to solve the leader election problem is perhaps the main source of hardness in our setting. One might also consider using {\em junta-driven} phase clocks, designed by G\k{a}sieniec and Stachowiak~\cite{DBLP:conf/soda/GasieniecS18}. This clock mechanism operates correctly in the range of parameters $\many X \in [1,n^{1-\delta}]$, which is the same requirement as in our clock mechanism, and also uses $O(1)$ states. However, the clock from~\cite{DBLP:conf/soda/GasieniecS18} will find itself stuck indefinitely in a central area of the phase space if the clock is initialized when $\many X=\Theta(n)$. When the value of $\many X$ is reduced, it will eventually start operating correctly with a period of $\Theta(\log n)$ rounds, but this will only happen after exponentially long time, in expectation.\footnote{This is precisely the reason why the solution to leader election from~\cite{DBLP:conf/soda/GasieniecS18} uses $\Theta(\log \log n)$ states per agent, and not $O(1)$ states.} This problem is completely alleviated in our work by building on the self-stabilizing oscillator design from~\cite{DBLP:journals/corr/DudekK17}, which leaves the central area of its state space in $O(\log n)$ rounds of the (sequential or parallel) scheduler driving the clock, in expectation.\footnote{The downside is that the clock we use is a little harder to manipulate, hence the solution to leader election and subsequent tasks becomes more involved in the layer of clock synchronization.}

The only remaining difficulty is that of designing a process which adapts the number of agents in the controlling state $\many X$ so that it is in the correct range $\many X \in n^{1-\eps}$ for a sufficiently long time to allow the clock hierarchy to organize itself and operate for a polylogarithmic number of rounds, i.e., for at least $\Omega(\log n)$ periods of the outermost clock. This is achieved by running a separate building block: a dynamical process which starts with all agents in a state representing $X$, and reduces $\many X$ over time. We denote the expected time from which $\many X \leq n^{1-\eps}$ is satisfied indefinitely in the protocol as $\cal T$; then, the time of convergence of  protocols formulated in the framework to a w.h.p. correct result will be given as  $\mathcal{T} + O(\poly \log n)$.

We consider two distinct ways of reducing $\many X$. For use in the always-correct framework, we consider such an auxiliary protocol with the simple rule of the form: \RULE{X}{X}{X}{\neg X}, which eliminates an agent in state $X$ whenever two agents in such a state meet. This guarantees that $\many X$ is non-increasing over time, that $\many X \geq 1$ is always satisfied, and that $\many X \leq n^{1-\eps}$ holds after ${\cal T} = O(n^{\eps})$ parallel time, w.h.p. (see Proposition~\ref{simpleelimination}).  For use in the faster w.h.p.\ framework, we use a slightly more involved $k$-level process which eliminates $X$ more quickly, and achieves  ${\cal T} = O(\log^k n)$, where $k$ is an arbitrarily chosen integer in the protocol design (see Theorem~\ref{k-stanowa-dynamika} for details). This approach will, however, result in an eventual disappearance of $X$ ($\many X = 0$ from some time step onwards), where $\many X > 0$ continues for $\Omega (\log^k n)$ after time $\cal T$, w.h.p. This is long enough for the protocol to successfully complete (w.h.p.), if $k$ is chosen suitably large with respect to the depth of the program formulation.

The designed execution framework comes with a number of guarantees which allow for the analysis of protocols formulated in it. In Section~\ref{le}, we put forward an appropriate leader election protocol and provide the corresponding analysis, whereas in Section~\ref{maj} we provide and describe a protocol for majority. Extending the protocol for leader election, the more general case of semi-linear predicates is handled in Section~\ref{general-semi-linear}.

In general, the w.h.p. versions of the protocols expressed in the programming language are relatively straightforward, and often constitute a simplification of previous designs with a super-constant number of states (e.g., the majority protocol mimics the solution known from~\cite{DBLP:conf/soda/AlistarhAG18}). Providing an always-correct variant of the protocols is a bit more involved, since some of the guarantees of correctness given by the primitives of the programming framework only hold w.h.p. (cf.~Theorem~\ref{th:guarantees}). A solution which is always correct is achieved by carefully combining a w.h.p. solution in the framework with a slower, deterministically correct solution running in parallel. This process of protocol combination is based on a notion of \emph{threads}, which are coupled (i.e., informally, executed asynchronously in parallel) in the framework. Threads may share some variables, and the specific interaction between the fast (main) thread and the deterministic thread is chosen in a separate (ad hoc) manner for each of the designed protocols, to allow for a proof of correctness of the composed protocol.

\paragraph{Relation to impossibility results.} The deterministically correct protocols which we present do not stand in contradiction to existing lower bounds on majority and leader election from Doty and Soloveichik~\cite{DBLP:conf/wdag/DotyS15}, Alistarh et al.~\cite{DBLP:conf/soda/AlistarhAEGR17} and Alistarh et al.~\cite{DBLP:conf/soda/AlistarhAG18}. Those impossibility results were derived (and stated) conditioned on an assumption on \emph{stable} computation --- to state it briefly, that once the protocol has reached a state that is a valid output, it remains in this state indefinitely, under any sequence of interactions (see~\cite{DBLP:conf/wdag/DotyS15} for a formal definition of stability). This seemingly safe assumption is in fact prohibitive, as it prevents, e.g., an approach of using fast protocol to quickly compute an output that is correct with high probability, and combining it with a slow and always correct protocol (running in the ``background'') to make sure that the computation is always correct, eventually. This is precisely the way we proceed to solve both leader election and majority in our framework.

When speaking of convergence time, we must emphasize that detecting whether a population protocol has converged or not is not possible locally, in any model. Indeed, for most basic tasks, such as majority or leader election, no agent can decide whether convergence has already been achieved or whether the outputs of some agents will subsequently change, due to the properties of the random scheduler which may isolate some subset of agents for an arbitrarily long time with positive probability, regardless of the applied protocol. 

\paragraph{Extensions of results.} We leave as open the question of whether always-correct finite-state protocols for the problems considered in this work converge in polylogarithmic time. In particular, we do not know if a solution exists to the following relaxed variant of the leader election problem: obtaining a finite-state protocol which creates in polylogarithmic time a junta having size $\many X \leq O(n^{1-\eps})$, while guaranteeing that the junta remains non-empty at (almost) all times if the dynamics are run forever. A solution to this problem appears to be a necessary basic building block for controlling all phase clock designs known to us, including the one considered in this work.

We also remark that in our solutions, after convergence to a correct output, the agents are still allowed to update their states in the part which is not used for encoding output. This is the case for our protocols, which continue running at least for some time after convergence. The time after which state changes in the population cease to happen, i.e., after which the protocol becomes silent, is $O(\poly\log n)$ for the w.h.p.\ schemes we present, whereas the always-correct schemes as presented in this paper never become silent. The latter solution can be modified to become silent after polynomial time (informally: once the deterministic thread has terminated with the same result as the main thread, for all agents), but the time after which the protocol becomes silent is still significantly longer than its $O(n^{\eps})$ convergence time. 

Finally, we note that in all our results are phrased in the randomized model of population protocols, in which agents have access to a constant number of fair coin tosses in each iteration, which they can use to select the transition rule in a given iteration. Phrasing the protocols to enforce deterministic operation is possible by simulating coin tosses from randomness of the fair scheduler, using the so-called synthetic coin technique~\cite{DBLP:conf/soda/AlistarhAEGR17}.

\subsection{Other Related Work}


\paragraph{Population Protocols.}

The population protocol model  captures the way in which the complex
behavior of systems (biological, chemical, sensor networks, etc.) emerges from the underlying local pairwise interactions
of agents. The original work of Angluin et al.~\cite{DBLP:journals/dc/AngluinADFP06, DBLP:journals/dc/AngluinAE08a}
 was motivated by applications in sensor mobility.
Despite the limited computational capabilities of individual sensors, population protocols permit the computation of two important classes of functions:
threshold predicates, which decide if the weighted average of types appearing in the population exceeds
a certain value, and modulo remainders of similar weighted averages. More precisely, the family of predicates computable in the finite-state population protocol model under the assumption of stability has been characterized as that of semi-linear predicates, or equivalently predicates expressible in second-order Presburger arithmetic~\cite{DBLP:journals/dc/AngluinADFP06}.

\paragraph{Majority and Leader Election.} The most common problems considered in the context of population protocols include \emph{majority} and \emph{leader election}. The majority problem is a special form of consensus~\cite{DBLP:conf/fct/Fischer83}, in which the final configuration reflects the unique color of the largest fraction of the  population initially colored with two colors.\footnote{The variant considered in this work is more general, since we allow some agents to be initially uncolored.} The majority problem was first posed in the context of population protocols in~\cite{DBLP:journals/dc/AngluinADFP06}
and later a 3-state protocol for \emph{approximate} majority was given in~\cite{DBLP:journals/dc/AngluinAE08}, which converges in $O(\log n)$ time, but requires that the population gap is $\Omega(\sqrt{n \log n})$. Draief and Vojnovi\'c \cite{DBLP:journals/siamco/DraiefV12} and later Mertzios et al. \cite{DBLP:conf/icalp/MertziosNRS14} considered a 4-state protocol for exact majority, however with a prohibitive polynomial convergence time ($O(n \log n)$ expected parallel time). Alistarh et al. \cite{DBLP:conf/podc/AlistarhGV15} were the first to provide a protocol for exact majority with polylogarithmic parallel time, however the number of states there can be polynomial if the initial gap is small enough ($O(n)$ states, $O(\log^2 n)$ time whp). Further studies on time-space trade-offs can be found
in Alistarh et al.~\cite{DBLP:conf/soda/AlistarhAEGR17} 
and Bilke et al.~\cite{DBLP:conf/podc/BilkeCER17} 
culminating with Alistarh et al.~\cite{DBLP:conf/soda/AlistarhAG18} showing a protocol with $O(\log n)$ states and $O(\log^2 n)$ expected time to elect a majority.

In the leader election problem, in the final configuration an unique agent must
converge to a {\em leader state} and every other agent has to stabilise in a {\em follower} state. A series of papers \cite{DBLP:conf/icalp/AlistarhG15}, \cite{DBLP:conf/soda/AlistarhAEGR17}, \cite{DBLP:conf/podc/BilkeCER17} culminated with Alistarh et al.~\cite{DBLP:conf/soda/AlistarhAG18} achieving a $O(\log n)$ states protocol electing a leader in $O(\log^2 n)$ expected time, improved by Berenbrink et al.~\cite{DBLP:conf/soda/BerenbrinkKKO18} to $O(\log^2 n)$ time with high probability. In a breakthrough paper, Gąsieniec and Stachowiak~\cite{DBLP:conf/soda/GasieniecS18} reduced the number of states exponentially to $O(\log \log n)$, and later Gąsieniec et al.~\cite{GSU18} achieved the same number of states but improved the time to $O(\log n \log \log n)$.


\paragraph{Progress on Phase Clocks.}

Unsurprisingly, the more efficient protocols in or around the population protocol framework~\cite{DBLP:conf/soda/BoczkowskiKN17,DBLP:conf/soda/GasieniecS18,DBLP:conf/soda/AlistarhAG18,GSU18} have focused on ways to allow some form of synchronization on the system to appear, in the form of a \emph{phase clock} or closely related construct. 
This line of work includes, in particular, {\em leader-less} phase clocks applied by Alistarh  et al. in \cite{DBLP:conf/soda/AlistarhAG18} and
{\em junta-driven} phase clocks used by G\k{a}sieniec and Stachowiak~\cite{DBLP:conf/soda/GasieniecS18}.

\subsection{Preliminaries and Notation}

Population protocols are expressed in the form of a set of rules, describing the state transitions of a pair of interacting agent. When designing protocols with $O(1)$ states, we use the convention that the state space of the agent is the Cartesian product of a certain number of boolean flags known as \emph{state variables}, which may be set or unset for each agent (we use the symbol $\on$ to denote the truth value and $\off$ to denote the false value). A rule can then be conveniently described through bit-masks, i.e., by specifying a set of $4$ boolean formulas $\Sigma_1, \Sigma_2, \Sigma_3, \Sigma_4$ on the state variables, written as follows: \RULE{\Sigma_1}{\Sigma_2}{\Sigma_3}{\Sigma_4}. Such a rule may be activated when the pair of interacting agents satisfy formulas $\Sigma_1$ and $\Sigma_2$ on their state variables, respectively. The execution of the rule corresponds to a minimal update of the states of the agents so that formulas $\Sigma_3$ and $\Sigma_4$ are satisfied by the states of the respective agents after the update. The special symbol $(.)$ is used to denote the empty boolean formula, which matches any agent. 

By convention, boolean variables associated with agents will be denoted by capital letters of the alphabet, such as $A$. The set of agents in the population for which $A$ is set will be denoted by $\mathcal{A}$. The number of agents in the population is denoted by $n$.

We apply a convention in which the scheduler picks exactly one rule uniformly at random from the set of rules of the protocol, and executes it for the interacting agent pair if it is matching. Protocols designed in this framework can be translated into frameworks in which all matching rules are executed systematically, e.g., in a top-down manner.

Our constructions rely on the idea of composing multiple protocols. In the simplest setting, this is obtained by defining each protocol with its own ruleset, and putting the rulesets of the different protocols together into one. We call protocols which have been composed like this \emph{threads}. (To ensure fairness of time division between the threads, mainly for the sake of better intuition of the reader, we will assume that each protocol in each thread is written down with the same number of rules; this can be enforced by creating a constant number of copies of the respective rules up to the least common multiple of the number of rules of respective threads.) We speak of \emph{composing protocol $P_2$ on top of protocol $P_1$} if the ruleset of protocol $P_2$ does not affect the values of boolean variables used in protocol $P_1$. Intuitively, the (asymptotic) execution of a protocol is not affected by composing a constant number of other protocols on top of it.

\section{Programming Framework for Protocol Formulation}

\subsection{Language Specification}
Our simple language for writing imperative code is based on the following constructs and assumptions. The code of the protocol is a collection of threads, sharing the same pool of boolean state variables. Variables are defined and initialized at protocol startup, either as $X \gets \on$ or $X \gets \off$.
The code of each thread is a finite-depth branching program with loops.
The only available control instructions are the following:
\begin{itemize}
\item \IFNONEMPTY{$condition$}\ [block] \ELSE\ [block] --- the branching instruction, where 'condition' is a boolean expression on local state variables.
\item \REPEAT\ [block] --- the outermost control loop of a thread.
\item \REPEAT[c \ln n]\ [block] --- possibly nested loops within a thread, where $c$ is an explicitly stated positive integer.
\end{itemize}
These constructs are intended to have an intuitive interpretation for the reader (which is indeed true in some circumstances w.h.p., as shown later). The intuition for the branching instruction corresponds to conditioning on the existence of at least one agent in the population for which the given boolean formula on local state variables evaluates to true.

The only primitive instructions are the following:
\begin{itemize}
\item \ASYNC[c \ln n]\ [ruleset], followed by a set of primitive rules.
\item The assignment instruction \GETS{X}{condition}, where $X$ is a variable and $condition$ is a boolean condition on local variables.
\item \ASYNC[c \ln n]\ [ruleset], followed by a set of primitive rules.
\end{itemize}
Intuitively, the assignment instruction is intended to update the states of all agents in the population to set $\mathcal{X}$ if and only if $condition$ is true, while the \ASYNC\ instruction is intended to run the provided set of rules on the population for the specified number of parallel rounds (allowing for nesting of population protocols). An execution of a rule from a specified ruleset or of a local variable assignment corresponding to the given assignment instruction for an agent in the population is called an \emph{operation}.

\subsection{Outline of Compilation and Execution Model}

An informal description of the compilation process of the code is as follows.  A certain (constant) number of threads are specified through code by the protocol designer. The remaining threads (also a constant number, dependent on the loop depth of the code formulation) are added internally to allow for clock operation and synchronization. For a program with $C_T$ threads, where $C_T$ is a constant, interacting agents pick a rule corresponding to the current step of each of the $C_T$ threads, choosing a thread u.a.r\ with probability $1/C_T$, independently of other interactions. For each thread, each agent runs its own control loop, which describes a finite state automaton. The automaton proceeds through alternating phases of ruleset execution, followed by synchronization. When executing a ruleset at depth $a$ in the loop structure of the code, its synchronization is governed by the $a$-th outermost clock of the system.  We formalize this description in Section~\ref{sec:compilation}.

\subsection{Compiled Population Protocols: Guarantees on Behavior and Convergence Time}

We will require that the compilation framework produces protocols which satisfy two constraints throughout its execution, known as the \emph{guaranteed behavior}.
\begin{definition}[Guaranteed behavior] The programming framework admits the \emph{guaranteed behavior property} if the following conditions are jointly fulfilled for any protocol compiled in it:
\begin{itemize}
\item At any time, for any agent, a state variable may only be modified in the way given by some primitive operation appearing in the code, i.e., by a rule in a ruleset or by an assignment operation.
\item Suppose that in some execution of the protocol, a state variable (or, more generally, a boolean formula on state variables) $S$ satisfies $\mathcal{S}(t) = \emptyset$ for all $t > t_0$. If at some time $t_1 > t_0$ an agent is executing an operation not contained in a branch of the form ``\IFNONEMPTY{S}'',  then this agent will never execute any operation contained in this branch in the future.
\end{itemize}
\end{definition}

In addition, we expect some further properties from the compiled protocol, which are to be met w.h.p. To describe them, we now introduce the following notation for an execution of the compiled protocol.
\begin{definition}[Synchronized iterations]
We say a protocol is \emph{synchronized} at a given moment of time if within a thread either all agents have their instruction pointer pointing to the same instruction location in each thread (all agents are active), or are in the process of entering or leaving the block with this instruction.

An \emph{iteration} is a time interval defined by looking at one agent (fixed arbitrarily at the start of the protocol) and considering the period of time from one time when this agent activates the first instruction of the outermost loop, to the next such time.

Finally, an iteration is said to be \emph{synchronized} if the protocol is synchronized at all moments of time during this iteration, all agents follow the same execution path, and moreover all instructions or rulesets contained within each internal \key{repeat} or \key{execute} instruction on the execution path of the program are executed (with all agents active) at least for the specified number of rounds. 
\end{definition}
We note that, in particular, for any \ASYNC[c \ln n]\ [ruleset] on the execution path instruction, in any synchronized iteration there must exist a period of $c\ln n$ rounds (or $c n \ln n$ steps), during which the program emulates an execution of a simple protocol consisting of set of rules $ruleset$ only under a \emph{fair} random scheduler. Note that, directly before and after this period, the given ruleset may also possibly be run for some further time for an arbitrary subset of agents, thus emulating the behavior of some \emph{unfair} scheduler on which no promises are made.

Our expectation that the programmed code does what it says it should is met in some synchronized iterations, which are called good iterations. 
\begin{definition}[Good iterations] An iteration is said to be \emph{good} if it is synchronized and additionally all executions of the assignment and \IFNONEMPTY{} instructions performed within the iteration reach their usual (w.h.p.) outcome for all agents, where:
\begin{itemize}
\item The expected outcome of a \GETS{X}{\Sigma} operation is that for each agent, the value of the state variable $X$ is set to the value of the boolean formula $\Sigma$ given on its local state variables.
\item The expected outcome of a \IFNONEMPTY{\Sigma} operation is that instructions contained in the ``\IF'' block are subsequently executed if and only if the boolean formula $\Sigma$ on local state variables evaluates to true for at least one agent in the population, and that instructions contained in the (optional) ``\ELSE'' block are otherwise executed.
\end{itemize}
\end{definition}

We are now ready to state the main Theorem on the properties of the framework. 
\begin{theorem}
\label{th:guarantees}
Any program expressed in the provided programming language can be compiled into a population protocol with $O(1)$ states, such that:
\begin{itemize}
\item[(i)] The guaranteed behavior constraints are always met;
\item[(ii)] Subject to a choice made at the time of compilation, one of the following claims holds:
\begin{itemize}
\item[(a)] Starting from the initialization of the protocol, after an initialization phase taking some number of rounds $O(\poly\log n)$, each of the $\Omega(\log n)$ iterations which follow directly afterwards is good and takes at most $O(\poly\log n)$ time; Or
\item[(b)] Starting from an arbitrary moment of time, after an initialization phase taking some number of rounds $O(n^{\eps})$, where $\eps>0$ is an arbitrarily chosen constant influencing the number of states of the compiled protocol, each of the $\Omega(\log n)$ iterations which follow directly afterwards is good and takes at most $O(\poly\log n)$ time. 
\end{itemize}
\end{itemize}
\end{theorem}
The claim of the Theorem follows from the compilation mechanism described in Section~\ref{sec:compilation} and the construction of the phase clock hierarchy described in Section~\ref{sec:hierarchy}.

\section{Warmup: Programming with High Probability of Correctness}

Using the framework is easiest when the goal is to achieve w.h.p.\ correctness of the result. In this scenario, it will, in most cases, be sufficient to create one Main thread only. The execution of the code can be seen as follows: for some time, the provided rulesets will be executed in no particular order. Then, w.h.p., some number of iterations of the outermost loop of the code will be executed as designed, i.e., respecting all rules of sequential programming, conditions, and loop limits. At some point during such correct execution of the program, after a sufficient number of iterations are completed w.h.p., it may slow (or stop) without warning.

Thus, the following use of the framework is safe and recommended: the Main thread of the code should be given as ``\REPEAT\ [Program]'', where ``Program'' is a piece of code which is required to solve the problem correctly in a sequential setting, subject to the two additional constraints: (1) ``Program'' does not modify any of the states of the input of the protocol (if any), and resets any other variables it may use to a valid initialization if it detects their initialization to be invalid; (2) If the output variables computed in one execution of ``Program'' were a valid answer to the problem, then the next execution of ``Program'' in the next iteration of the outer ``\REPEAT'' loop should not alter the state of output variables for any agent.

In the above, constraint (1) eliminates the problem with the initial uncontrolled phase of the execution before its correct operation starts, and (2) handles the issue with the program stopping at an unpredictable moment (e.g., when writing output variables).

To illustrate this approach, we provide the programs \fun{\Majority} and \fun{\LeaderElection} which describe protocols solving the respective problems w.h.p. of correctness.

\subsection{Leader Election Protocol (w.h.p.)}
\label{le}

\begin{theorem}
\label{lewhp}
Let $T = \Omega(\log n)$ be fixed arbitrarily. At the end of any $T$-th successive good iteration, \fun{\LeaderElection} has elected a unique agent with a set boolean state variable $L$, w.h.p. 
\end{theorem}

\begin{figure}[!!!ht]
\begin{protocol}
\PROTOCOL{\LeaderElection}
\VAR{\OUTPUT{\INIT{L}{\on}}}

\THREAD{\Main}{L}
    \LOCAL{\INIT{D}{\off}, \INIT{F}{\on}}
    \REPEAT
        \IFNONEMPTY{L}
            \GETS{F}{\{\on, \off\} $chosen uniformly at random$}
            \GETS{D}{L \wedge F}
            \IFNONEMPTY{D}
                \GETS{L}{D}
        \ELSE
            \GETS{L}{\on}
\end{protocol}
\end{figure}

\begin{proof}
Consider the first good iteration of the Main thread. Assume $\cal L$ is nonempty. Denote by $\ell_0,\ell_1,\ldots$ the size of $\cal L$ in the $i$-th good iteration after the first. We have $\mathbb{E}[\ell_{i+1} | \ell_i] = \ell_i/2 + 2^{-\ell_i} \cdot \ell_i$, so for $\ell_i \ge 2$ it follows that $\mathbb{E}[\ell_{i+1} | \ell_i] \le \frac34 \ell_i$. By the multiplicative drift theorem~\cite{DBLP:journals/algorithmica/DoerrJW12}, it follows that for some $t = c \log n$ for large enough constant $c$, $\ell_t = 1$ with high probability. We then also have $\ell_T = \ell_t = 1$ for any $T \geq t$.

If $\cal L$ is empty at the beginning of the first good iteration, then in the next iteration we have $|{\cal L}| = n$, and subsequently the same analysis applies.
\end{proof}

It follows that the convergence time of \fun{LeaderElection} is $O(\log^2 n)$ parallel rounds w.h.p. under the compilation scheme from Theorem~\ref{th:guarantees}(ii)(a), as each iteration of the protocol with no nested loops can be realized in $O(\log n)$ parallel rounds.

\subsection{Majority Protocol (w.h.p.)}
\label{maj}

\begin{theorem}
\label{th:majority_runtime}
Let $T = \Omega(\log n)$ be fixed arbitrarily. At the end of any $T$-th successive good iteration, \fun{\Majority} has computed in the boolean state variable $Y_A$ a correct answer to the majority problem on sets $\cal A$ and $\cal B$, w.h.p.
\end{theorem}

\begin{figure}[!!!ht]
\begin{protocol}
\PROTOCOL{\Majority}
\VAR{\OUTPUT{Y_A}, \INPUT{A, B}}

\THREAD{\Main}{Y_A, \READONLY{A,B}}
    \LOCAL{\INIT{A^*}{\off}, \INIT{B^*}{\off}, \INIT{K}{\off}}
    \REPEAT
        \GETS{A^*}{A}
        \GETS{B^*}{B}
        \REPEAT[c \ln n]
            \ASYNC[c \ln n]
                \RULE{A^*}{B^*}{\neg A^*}{\neg B^*}
            \GETS{K}{\off}
            \ASYNC[c \ln n]
                \RULE{A^* \wedge \neg K}{\neg A^* \wedge \neg B^*}{A^* \wedge K}{A^* \wedge K}
                \RULE{B^* \wedge \neg K}{\neg A^* \wedge \neg B^*}{B^* \wedge K}{B^* \wedge K}
        \IFNONEMPTY{A^*}
            \GETS{Y_A}{\on}
        \IFNONEMPTY{B^*}
            \GETS{Y_A}{\off}
\end{protocol}
\end{figure}

\begin{proof}
We follow the steps of \cite{DBLP:conf/soda/AlistarhAG18}, adapting the proof to our setting.
Suppose w.l.o.g.\ that initially $|\set A| < |\set B|$.
Thread \fun{\Trim} maintains the invariant $|\set A| - |\set B|$. Consider the first good iteration. All of the following properties hold with high probability, conditioned on $c$ being large enough constant. ${A}^*$ and ${B}^*$ are initialized so that $|\mathcal{A}^*| < |\mathcal{B}^*|$. Consider a single loop-iteration of the loop in line 9 (that is, single pass through lines 10--15). Denote by $a_i$ the size of $\mathcal{A}^*$ at the start of $i$-th loop-iteration, by $a'_i$ the size of $\mathcal{A}^*$ in the $i$-th loop-iteration at line 12, and denote $b_i$ and $b'_i$ in the same manner with respect to $\mathcal{B}^*$. 

We observe the following: if for some $i$, in the $i$-th loop-iteration there is  $b'_i \ge \frac16n$, then in that loop-iteration, in lines 10--12 $|\mathcal{B}^*|$ is always at least $\frac16n$, and any $x \in \mathcal{A}^*$ has at least $1 - (1 - \frac{1}{6n})^{c n \ln n} = 1 - n^{-c/6}$ probability of triggering rule from line 11. As a consequence, $a'_{i}=0$.

Now, assume that for some $i$ there is $a'_i,b'_i \le \frac16n$. In such case $a_{i+1} \le 2a'_i$ and $b_{i+1} \le 2b'_i$ and there are always at least $\frac13n$ nodes that do not belong to $\mathcal{A}^*$, $\mathcal{B}^*$.
Additionally, any node that at line 12 belonged to $\mathcal{A}^*$ has probability at each step least $\frac{1}{3n}$ of triggering rule from line 14. Thus with high probability, it triggers it exactly once. The same reasoning holds for $\mathcal{B}^*$, and as a consequence, $a_{i+1} = 2a'_i$ and $b_{i+1} = 2b'_i$.

It now follows that for some $t \le O(\log n)$, it holds that $a'_t= 0$. Otherwise, we would have that for any $i \le t$, we have $b'_i \le \frac16n$. Thus $b'_i - a'_i = b_i - a_i = 2^{i-1} (b_1 - a_1)$, and $b'_t - a'_t > n$, a contradiction.

It follows that after leaving loop from lines 9--15, $\mathcal{A}^* = \emptyset$ and $\mathcal{B}^* \not= \emptyset$, thus $\set Y_A = \emptyset$ holds.
\end{proof}

It follows that the convergence time of \fun{LeaderElection} is $O(\log^3 n)$ parallel rounds w.h.p. under the compilation scheme from Theorem~\ref{th:guarantees}(ii)(a), since each iteration of the protocol with a single nested loop can be realized in $O(\log^2 n)$ parallel rounds.

\section{Compilation Process of the Sequential Code}\label{sec:compilation}

The proposed language for expressing sequential code is first precompiled into a subset of the language with a simple tree grammar. All leaf nodes take the form of \ASYNC[c \ln n]{\ [ruleset]} instructions with provided rulesets, and all internal nodes take the form of \REPEAT[c \ln n]{\ [ChildNode1, ChildNode2, \ldots]} instructions (with the exception of the root node, which is typically given by a \REPEAT{} loop). (If different values of $c$ were specified for loopcounts in the description of the code, w.l.o.g.\ we understand $c$ to be the maximum among them, and use this value of $c$ throughout the code.)

We describe below the process of elimination of the other language constructs. We start by precompiling all assignment operation by replacing them with a simple ruleset, and then evaluate all conditions and eliminate the branching structure from the program.

\paragraph{Assignments.}
A naive but correct way of implementing an assignment operation \GETS{X}{\Sigma}, for some boolean formula $\Sigma$ on local state variables, is through an insertion of code in the form of loops shown in Fig.~\ref{fig:assignment}, using an auxiliary internal trigger state variable $K_{(\ID\ )}$ assigned to the line number $\#$ in which the instruction is placed. The assignment operation ensures that in any circumstances, if the value of $X$ changes, then $X$ may become set only when $\Sigma$ is set and $X$ may become unset only when $\Sigma$ is unset. Moreover, under correct operation, the trigger $K_{(\ID\ )}$ will be set for each agent at the beginning of the second loop w.h.p., and unset while performing the assignment of $X$ for each agent; thus, the assignment will be performed at most on each agent, and exactly once w.h.p. (where the high probability may be made arbitrarily high through a careful choice of $c$).

\begin{figure}[h]
\begin{protocol}
\ASYNC[c \ln n]
    \RULE{\neg K_{(\ID )}}{.}{K_{(\ID )}}{.}

\ASYNC[c \ln n]
    \RULE{\Sigma \wedge K_{(\ID )}}{.}{X \wedge \neg K_{(\ID )}}{.}
    \RULE{\neg\Sigma \wedge K_{(\ID )}}{.}{\neg X \wedge \neg K_{(\ID )}}{.}
\end{protocol}
\vspace{-5mm}
\caption{Precompilation of instruction ``\GETS{X}{\Sigma}'' placed in line number $\#$.}\label{fig:assignment}
\end{figure}

\paragraph{Conditions and branching.} A conditional statement following instruction ``\IFNONEMPTY{X}'' has its condition evaluated following an insertion of code presented in Fig.~\ref{fig:ifnonempty}, using an auxiliary flag variable $Z_{(\ID\ )}$ assigned to the line number $\#$ in which the instruction is placed. In the code from Fig.~\ref{fig:ifnonempty}, in the first line flag $Z_{(\ID\ )}$ is unset (w.h.p.) for all agents (and never become set for any agent). In the loop which follows, an epidemic process is triggered with source $X$ on flags $Z_{(\ID\ )}$. In this phase, a flag $Z_{(\ID\ )}$ may become set for an agent only if it is already set at the same time for at least one other agent in the population, or if set $\mathcal X$ is nonempty, i.e., if $\mathcal{Z}_{(\ID\ )} = \emptyset$ at some time, then $\mathcal{Z}_{(\ID\ )}$ may become nonempty only if $\mathcal X$ is nonempty at some point. Moreover, if set $\mathcal X$ is nonempty throughout the loop, then at the end of the loop, flag  $Z_{(\ID\ )}$ will be set for all agents w.h.p.

We note the behavior of the code below when set $\mathcal X$ becomes permanently empty from some time $t_0$ onwards. Then, at the end of the first correct execution of the code evaluating ``\IFNONEMPTY{X}'', we will have flag $Z_{(\ID\ )} = \off$ for all agents. Suppose this happens at time $t_1 > t_0$. Then, for all $t > t_1$, no flag $Z_{(\ID\ )}$ will ever be set again, regardless of the correctness of the execution of the protocol. This will ensure the desired property of correctness of operation.

\begin{figure}[h]
\begin{protocol}
\VAR{\OUTPUT{Z_{(\ID )}}}

\GETS{Z_{(\ID )}}{\off}
\ASYNC[c \ln n]
    \RULE{X}{.}{Z_{(\ID )}}{Z_{(\ID )}}
    \RULE{Z_{(\ID )}}{.}{Z_{(\ID )}}{Z_{(\ID )}}
\end{protocol}
\vspace{-5mm}
\caption{Precompilation of instruction ``\IFNONEMPTY{X}'' placed in line number $\#$. }\label{fig:ifnonempty}
\end{figure}

After running the evaluation condition, intuitively, we require that the agent executes operations within the block of the corresponding ``\IF'' branch only if $Z_{(\ID\ )} = \on$, and executes operations within the block of the ``\ELSE'' branch only if $Z_{(\ID\ )} = \off$. The precompiled form of the code does \emph{not}, however, have branching structure. We perform a standard operation of compacting rulesets of the ``\IF'' and ``\ELSE'' blocks into a single ruleset, adding a requirement on $Z_{(\ID\ )}$ or $\neg Z_{(\ID\ )}$, respectively, to the conditions requiring for triggering rules from the respective block. Formally, this is done in a bottom-up manner, and when compacting the ``\IF'' and ``\ELSE'' blocks corresponding to a condition in line $\#$, we assume that all branching instructions contained within these two blocks of code have already been eliminated. Thus, the only instructions which remain in these blocks of code are assumed to already have the required structure of a tree program with \ASYNC[c \ln n]{\ [ruleset]} instructions at leaf nodes and \REPEAT[c \ln n]{\ [ChildNode1, ChildNode2, \ldots]} instructions at internal nodes. We denote the trees corresponding to the respective blocks $T_{if (\ID\ )}$ and $T_{else (\ID\ )}$. We construct one single tree  $T_{(\ID\ )}$ out of them as follows. First, we augment each of the trees until they are isomorphic by inserting artificial repeat loops and nil instructions (empty rulesets) in the artificially created leaves. The tree  $T_{(\ID\ )}$ is set to be isomorphic to each of these augmented trees. There is now a one-to-one matching between the leaves of the trees $T_{if (\ID\ )}$ and $T_{else (\ID\ )}$. Now, for the $i$-th leaf of $T_{(\ID\ )}$ we define a ruleset $R^i$ by taking a union of modified rules from rulesets in the $i$-th leaf of $T_{if (\ID\ )}$ (ruleset $R^i_{if}$) and from $T_{else (\ID\ )}$ (ruleset $R^i_{else}$) as follows:
\begin{itemize}
\item For each rule of the form \RULE{\Sigma_1}{\Sigma_2}{\Sigma_3}{\Sigma_4} in ruleset $R^i_{if}$, we put in ruleset  $R^i$ the rule \RULE{Z_{(\ID\ )} \wedge \Sigma_1}{Z_{(\ID\ )} \wedge \Sigma_2}{\Sigma_3}{\Sigma_4}.
\item For each rule of the form \RULE{\Sigma_1}{\Sigma_2}{\Sigma_3}{\Sigma_4} in ruleset $R^i_{else}$, we put in ruleset  $R^i$ the rule \RULE{\neg Z_{(\ID\ )} \wedge \Sigma_1}{\neg Z_{(\ID\ )} \wedge \Sigma_2}{\Sigma_3}{\Sigma_4}.
\end{itemize}
The above is repeated until all branching elements have been eliminated from the code.

\paragraph{Precompilation result.}

At the end of the process, we obtain a tree $T$ with (loop) depth $l_{\max}$ and (loop) width $w_{\max}$. Each internal node represents a loop, and each leaf node a ruleset, which must be repeated for at least $c \ln n$ rounds. Without loss of generality, we assume that all internal nodes have the same number of children $w_{\max}$; assuming this corresponds to padding the tree $T$ into a complete ($w_{\max}$)-ary tree of depth $l_{\max}$, inserting artificial repeat loops and nil instructions (empty rulesets) in the artificially created leaves.

The following Section describes how to convert the obtained structure into specific rules triggered in a way synchronized by a hierarchy of phase clocks (see Subsection~\ref{sec:compiledrules} for the compilation mechanism itself).

\section{Construction of The Phase-Clock Hierarchy}\label{sec:hierarchy}

\subsection{Clocks with Arbitrary Constant Period}\label{sec:clocks}

In what follows, lengths of time intervals are expressed in parallel rounds.

A \emph{clock} is a protocol with states ($C_{0}, \ldots, C_{m-1}$), for some positive integer $m$ called its \emph{module}, such that in the course of execution each agent performs a sequence of state transitions, moving it from some state $C_i$, $i \in \{0,\ldots, m-1\}$, to state $C_{(i+1) \bmod m}$. The time between two successive moments when an agent enters state $C_0$ is called its \emph{clock cycle}. We say that a clock is \emph{operating correctly at tick length} $r \geq \ln n$ during a time interval $I$ if we can find an ordered sequence of subintervals $I_i^j \subseteq I$, $i \in \{0,\ldots, m-1\}$, $j \in \{1,\ldots j_{\max}\}$ for some integer $j_{\max}$, called \emph{ticks}, such that:
\begin{itemize}
\item $\max I_{i}^j < \min I_{i+1}^j$ for all indices $j$ and $i < m-1$, and $\max I_{m-1}^j < \min I_{0}^{j+1}$ for all indices $j < j_{\max}$
\item Throughout each tick $I_i^j$, all agents in the population are in the same state $C_i$,
\item All ticks have length $|I_i^j| \in [a r, b r]$, where $a > 0$ can be fixed arbitrarily, and $b > a$ is a protocol-specific constant depending on the choice of $a$;
\item Adjacent ticks are not too far apart: the set of intervals between ticks $I \setminus \bigcup_{i,j} I_i^j$ does not contain an interval of length more than $b r$.
\end{itemize}
(We remark that the tick length of a correctly operating clock asymptotically determines the cycle length of all of the agents.)
All algorithms in this paper are self-contained, except that we reuse clock routines known from the literature. We are making use of clock protocols which, after a brief initialization phase operate correctly over long intervals of time (i.e., over a sufficiently large polylogarithmic or polynomial number of rounds), w.h.p.

A clock protocol $C'$ with a module $m' \geq 3$ (where in what follows we will use a clock with $m'=3$) and states ($C'_{0}, \ldots, C'_{m'-1}$), operating with ticks of length $|I_i^j| \in [a' r, b' r]$, can be used to provide a clock with some longer module $m > m'$ and states ($C_{0}, \ldots, C_{m-1}$). We do this by composing clock protocol $C'$ with the following set of rules:

\RULE{C'_0}{C'_0}{C'_0 \wedge T}{C'_0}

\RULE{C'_1  \wedge T \wedge C_i}{C'_1}{C'_1 \wedge \neg T \wedge C_{(i+1) \bmod m} \wedge \neg C_i}{C'_1},\quad for all $i \in \{0, \ldots, m-1\}$

\RULE{C'_2 \wedge C_i}{C'_2 \wedge C_j}{C'_2 \wedge C_i}{C'_2 \wedge C_i \wedge \neg C_j},\quad for all $j < i$, $i,j \in \{0, \ldots, m-1\}$.

We observe that if clock $C'$ operates correctly with constants $a', b'$ during time interval $I$ of given length, $|I| = O(\poly(n))$, where $a'$ is chosen to be sufficiently large, then clock $C$ operates correctly during this time interval with constants $a = a' (m'-1)$ and $b = c b' m'$, w.h.p., for some choice of constant $c$ depending on the required probability of correctness.  Indeed, note that the first two rules of the composition activate a trigger $T$, which is activated for each agent at most once per cycle of clock $C'$, and is activated for all agents exactly once in all cycles of clock $C'$, w.h.p.\ (because the length of each tick of clock $C'$ is at least $a' \ln n$, for a sufficiently large choice of constant $a'$). An agent with a set trigger $T$ in a given clock cycle advances its state for clock $C$ within the first $O(\ln n)$ rounds of tick $C'_1$ of the current clock cycle, following the second rule. The third rule ensures that during tick $C'_2$, all agents w.h.p. agree on the same state of clock $C$, chosen as the maximum of all states of $C$ over all agents in the population. Note that once agreement on some state $C_i$ is reached during a given cycle of clock $C'$, then such agreement will be retained throughout all future ticks $C'_2$ of clock $C'$ during interval $I$ w.h.p., with the state of clock $C$ in the population advancing from some $C_i$ to $C_{(i+1) \bmod m}$ in each cycle of clock $C'$.

Thus, clock $C$ is ``slower'' by a factor of $m'$ with respect to clock $C'$. In our construction, we will use as the base clock the simple clock $C'$ with $m'=3$ described in Section~\ref{sec:baseclock}, 
and use it to generate a clock $C$ with a significantly larger module $m$, where $m$ will depend on the sequential code being executed as $O(w_{\max})$.

\subsection{Design of the Base Clock}\label{sec:baseclock}

We first proceed to describe the design of the base modulo-3 phase clock protocol, which meets the requirements for the clock hierarchy laid out in Section~\ref{sec:clocks}. All other protocols are then composed with this clock. We base the clock on the design of the oscillator protocol $P_o$ described in~\cite{DBLP:journals/corr/DudekK17}. We recall briefly its basic properties.

Protocol $P_o$ uses $7$ states: six states of the oscillator, called $A_i^+$ and $A_i^{++}$, for~$i\in\{1,2,3\}$, and an optional control (source) state, denoted by $X$. Each agent of the oscillator holds one of the three \emph{species} $A_i := A_{i}^{+} \vee A_{i}^{++}$, for $i\in\{1,2,3\}$. (The protocol $P_o$ is itself inspired by the so-called rock-paper-scissors oscillator dynamics, which is defined by the simple predator-prey rule ``\RULE{A_i}{A_{(i-1)\bmod 3}}{A_i}{A_i}'', for $i=1,2,3$; in protocol $P_o$, this rule works with slightly different probability for the states $A_i^+$ and $A_i^{++}$ within species $A_i$.). We denote $a_{\min}= \min_{i=1,2,3}|\mathcal{A}_i|$.
The control state converts any encountered agent of some species $A_i$ to an agent of a uniformly random species.

The theorem below is obtained by carefully retracing the arguments in~\cite{DBLP:journals/corr/DudekK17}[Section 6.4,Section 7.1]\footnote{We remark that the proof of claim~(i) of the Theorem essentially follows from \cite{DBLP:journals/corr/DudekK17}[Section 6.4; Lemmas 1-8], but the analysis therein was conducted for $|\mathcal X|=0$. Assuming $\mathcal X \leq n^{1-\eps}$ instead adds additional terms to the considered potentials without changing the proof structure}. 
\begin{theorem}[variant of analysis in \cite{DBLP:journals/corr/DudekK17}]
\label{th:variantofclock}
Fix $0 < \eps < 1/2$. If $|\mathcal{X}| \in [1, n^{1-\eps}]$ at all time steps starting from some round $t_0$, then regardless of the configuration at time $t_0$, the following properties hold:
\begin{itemize}
\item[(i)] the system reaches a configuration with $a_{\min} < n^{1- \eps/2}$ at some time $t_1 > t_0$, where $t_1 - t_0= O(\log n)$ rounds in expectation.
\item[(ii)] for any moment of time $t \in [t_1, t_1 +e^{\Omega(n)}]$ we have $a_{\min} < n^{1-\eps/3}$ w.h.p., and for each species index $i=1,2,3$, there exists a time step in the interval of rounds $[t, t+O(\log n)]$ when at least $n - O(n^{1-3\eps})$ agents belong to species $A_i$. Moreover, if $A_i$ is at the current time step the species in the population held by at least $n - O(n^{1-3\eps})$ agents, then the next species with this property will be $A_{i+1}$, w.h.p. 
\end{itemize}
The above properties hold under the assumption of an asynchronous fair scheduler or a random-matching fair synchronous scheduler.\qed
\end{theorem}
Condition (ii) in the above theorem characterizes the oscillatory behavior of the protocol once set $\mathcal{X}$ is small enough (but non-empty), after a short period of convergence given by condition (i). We remark that condition (i) holds because $\mathcal {X}$ is not too large, whereas condition (ii) requires $\mathcal X$ to be both strictly positive and not too large. Under these assumptions, each species then regularly switches from being the smallest (held by at most $O(n^{c'})$ agents where $c' := 1 - \eps/3$) to being the largest (held by all but $o(n)$ agents), with the largest species alternating in the cyclic order $\ldots \to A_1 \to A_2 \to A_3 \to A_1 \ldots $, with each cycle and each phase in it taking $\Theta(\log n)$ steps. This oscillatory effect provides the most important components of a clock, but does not yet implement the separation between ticks required in Section~\ref{sec:clocks}. 

We are now ready to build a phase clock.
Consider protocol $P_o$ composed with the following ruleset on the states $C'_s$, where $s\in \{0,\ldots, 3k-1\}$, and $k = \Omega(1/c')$ is a sufficiently large constant:

\RULE{C'_{s}}{A_{i+1}}{C'_{(s+1)\bmod 3k} \wedge \neg C'_{s}}{A_{i+1}}

\RULE{C'_{s}}{\neg A_{i+1}}{C'_{k\cdot i}\wedge \neg C'_{s}}{\neg A_{i+1}}

for $i = \lfloor s/k \rfloor  \in\{0,1,2\}$. We call such protocol a \emph{clock controlled by an external signal} $\mathcal{X}$. We add an auxiliary states $C_i := \bigvee_{r \in\{1,\ldots, k\}}C'_{ki+r}$, for $i=1, 2, 3$. The boolean formulas $C_1, C_2, C_3$ (which can also be made explicitly into state variables) now provide the required interface to $3$-state clock. We denote this composition of protocols as $C_o$.

\begin{theorem}
Assume protocol $C_o$ is initialized so that $a_{\min} < n/10$ and that there is time $t_0,t_1$ such that for any $t \in [t_0,t_1]$ there is $0 < |\mathcal{X}| < n^c$ for some constant $c<1$. Then $C_o$ is a clock operating correctly during a time interval $[t_0,t_1]$, w.h.p.
\end{theorem}
\begin{proof}
The above rules move the state of each agent around a cycle of length $3k$. An agent in state $C'_{(ki) \bmod 3k}$ (which may be interpreted as ``believing'' $A_{i+1}$ to be a minority species at the current time) waits in states of the form $C'_{(ki+r) \bmod 3k}$, for $r < k$, until in $k$ consecutive meetings it has met agents from species $A_{i+1}$ only. This implies, w.h.p., that species $A_{i+1}$ is represented by a sufficiently large fraction of the population $\omega(n^{c'})$, and the state moves to $C'_{(k(i+1)) \bmod 3k}$, where the appearance of $k$ consecutive meetings with agents from species $A_{i+2}$ is awaited, etc. The rate of traversal of this cycle corresponds to the rate of oscillatory behavior of protocol $P_o$, moreover, after one cycle of the oscillator $P_o$, all agents in the population will become synchronized (up to a possible shift of $\eps \ln n$ rounds, w.h.p., where $\eps$ can be made arbitrarily small) as to the values.
\end{proof}

It remains to discuss how to ensure that protocol $P_o$ displays its oscillatory behavior through appropriate initialization.

\paragraph{Controlling $|\mathcal X|$ for always-correct protocols.}

We ensure that the number of agents with set control state $X$ satisfies the required condition $0 < |\mathcal{X}| < n^{1-\eps}$, for any $\eps > 0$, quickly and perpetually after the initialization of the protocol.
\begin{proposition}
\label{simpleelimination}
There is a protocol using $O(1)$ states with one marked state $X$ such that $|\mathcal{X}|>0$ is guaranteed and in $O(n^{1-\varepsilon})$ time $|\mathcal{X}| < n^\eps$, for any $\eps > 0$, w.h.p.
\end{proposition}
\begin{proof}
This can be achieved within $O(n^\eps)$ rounds, w.h.p., by initializing $X \gets \on$ for all agents, and applying the following rule (in composition with all other rules) \RULE{X}{X}{\neg X}{X}
which eventually unsets $X$ for all but one agents. A folklore computation following from a concentration analysis around the governing equation $\mathbb{E}\frac{d |\mathcal{X}|}{dt} = - (|\mathcal{X}|/n)^2$, gives the required bound on the number of rounds after which the bound  $|\mathcal{X}| < n^{1-\eps}$ is satisfied. Indeed, consider $T_j$, the number of interactions it takes for $|\mathcal{X}|$ to drop from $n/2^{j}$ to $n/2^{j+1}$, call it phase $j$. Since each interaction decreases $|\mathcal{X}|$ with probability between $(1/2^{j})^2$ and $(1/2^{j+1})^2$, we have that with very high probability $T_j = \Theta(n \cdot 2^j)$ for $j$ such that $n/2^j \ge n^{1-\eps}$. By observing that $\sum_{j=0}^{\eps\log_2 n} T_j = \Theta(T_{\eps\log_2 n}) = \Theta(n \cdot n^{\eps})$, we have that the desired behavior is achieved after $O(n^{\eps})$ rounds.
\end{proof}

We remark that achieving the bound $|\mathcal{X}| < n^{1-\eps}$ is possible more quickly by using a super-constant number of auxiliary states. This task is known as \emph{junta-election}:
\begin{proposition}[c.f. \cite{DBLP:conf/soda/GasieniecS18} Lem.\ 4.2,\ 4.5, Thm.\ 4.1]
There is a protocol using $O(\log \log n)$ states with one marked state $X$ such that $|\mathcal{X}| > 0$ is always guaranteed, and in $O(\log n)$ time, $|\mathcal{X}| < n^{1-\eps}$, for any $\eps > 0$, w.h.p.
\end{proposition}

\paragraph{Controlling $|\mathcal X|$ for w.h.p. protocols.} When w.h.p.\ results are required, it is possible to adapt protocol $P_o$ to perform $\Theta(\poly \log n)$ good oscillations before stopping, w.h.p., while using only $O(1)$ states.

In fact, to achieve the correct number of oscillations, it suffices to construct a ``signal'' following the asymptotic relation $|\mathcal{X}| \sim \exp(-t^{1/k})$ over time $t$, for some constant $k$.

\begin{proposition}
\label{k-stanowa-dynamika}
There is a protocol using $O(1)$ states with one marked state $X$ such that  in $O(\poly \log n)$ time $|\mathcal{X}| < n^{1-\eps}$, for any choice of $\eps>0$, w.h.p.
\end{proposition}
\begin{proof}





We observe that for $k>0$, the same can be achieved by using a higher-order variant of the previously analyzed process, where for some constant $k\ge 1$, we initialize $Z \gets \on$ and additional state flags $Z_1,\ldots,Z_{k} \gets \off$.


\RULE{.}{\neg Z}{\neg Z_1 \wedge \ldots \wedge \neg Z_k}{.}

\RULE{Z \wedge \neg Z_1 \ldots \wedge \neg Z_k}{Z}{Z_1}{.}

\RULE{Z_i}{Z}{\neg Z_i \wedge Z_{i+1}}{.} for any $i<k$

\RULE{Z_k}{Z}{\neg Z \wedge \neg Z_k}{.}

This protocol implements dynamics of the form $\mathbb{E} \frac{d |\mathcal{Z}|}{dt} = - |\mathcal{Z}| \cdot (|\mathcal{Z}|/n)^{k+1}$, which solves to $|\mathcal{Z}| = \Theta( n \cdot t^{-1/(k+1)})$.

We now use $Z$ to create a signal $X$, initially $X \gets \on$, using additional states $X_1, \ldots, X_{k-1} \gets \off$.

\RULE{.}{\neg Z}{\neg X_1 \wedge \ldots \wedge \neg X_k}{.}

\RULE{X \wedge \neg X_1 \ldots \wedge \neg X_g}{Z}{X_1}{.}

\RULE{X_i}{Z}{\neg X_i \wedge X_{i+1}}{.} for any $i<k-1$

\RULE{X_{k-1}}{Z}{\neg X \wedge \neg X_{k-1}}{.}

This implements a dynamics of the form $\mathbb{E} \frac{d |\mathcal{X}|}{dt} = - |\mathcal{X}| \cdot (\mathcal{Z}/n)^{k} = - |\mathcal{X}| \cdot t^{-k/(k+1)} $, which solves to $|\mathcal{X}| = n \cdot e^{-t^{1/k}}$.

The fact that the required dynamics is in fact realized with high probability by the given protocol follows from a mean-field analysis on the continuous process and standard concentration bounds (notably, taking into account~\cite{DBLP:journals/corr/DudekK17}[Lemma~1] and using Azuma's inequality).
\end{proof}

\subsection{Hierarchy of Clocks with Logarithmically Slowed Rate}
\label{hierarchia}

In what follows, we will be using a base clock $C^{(j)}$, where $j \geq 1$, with rate $r^{(j)}$ and sufficiently long module $m$, where we assume $4|m$ for technical reasons to simulate a slowed down version of other protocols. (For clock $C^{(1)}$, we put $r^{(1)} = \alpha \ln n$ for some large constant $\alpha$ depending on the sequential code being executed.) The simulation of a protocol $P$ proceeds as follows. All agents have state variables of clock $C^{(j)}$ and two copies of state variables of protocol $P$, which we call the \emph{current} copy and a \emph{new} copy. Then, in composition with the rules of clock $C^{(j)}$, the following rules are executed for all $i \in \{ 0, \ldots, m/4-1\}$:
\begin{enumerate}
\item When two agents meet in state $C^{(j)}_{4i}$, both having a set trigger $S$, then they simulate the interaction of protocol $P$ on the current copy of the state variables of protocol $P$, storing the new values of their state variables to the respective new copies, and unsetting trigger $S$ in the same interaction.
\item When two agents meet in state $C^{(j)}_{4i+2}$, they assign the new copy of their state variables for protocol $P$ to the current copy, and set trigger $S$ in the same interaction.
\end{enumerate}
We first observe that when clock $C^{(j)}$ is operating correctly, rules of the form 1 and 2 above are separated in time by the odd clock ticks of clock $C^{(j)}$. Thus, w.h.p. the operation of the simulation of protocol $P$ can be more abstractly viewed as a computation of a random linear-size matching $M$\footnote{Matching $M$ is created incrementally, with the first agents to interact after entering state $C^{(j)}_{4i}$ being more likely to join it. However, for the applied clocks which are based on missing species detection, the probability that an agent enters state $C^{(j)}_{4i}$ at a given time is almost the same for all agents, differing for any two agents by at most $O(n^{-k})$, where constant $k$ can be made arbitrarily large by a choice of constant $a$ of the clock. It follows by an elementary analysis that the probability that $M$ is any fixed matching of given size is within a factor of $(1\pm O(n^{-k+2}))$ from that in the uniform distribution on the set of all matchings of the same size. This polynomially small non-uniformity becomes negligible in protocol analysis.} on the population in between the $(4i-1)$-st and $(4i+1)$-st clock ticks of the current clock cycle, and then an execution of the corresponding rules of protocol $P$ along the chosen matching $M$, \emph{deferred} to some time between the $(4i+1)$-st and $(4i+3)$-rd clock ticks of the current clock cycle. We have thus effectively (i.e., conditioned on the non-occurrence of events occurring with probability $O(n^{-c})$ where $c$ can be made an arbitrarily large constant) implemented a fair execution of protocol $P$ under a random-matching scheduler, but slowed by a factor of $\Theta(r^{(j)})$ with respect to the natural execution rate of such a scheduler.

Whereas we are unaware of any generic method of showing the equivalence between a random-sequential and random-matching scheduler (and consider designing such a method to be a question of significant interest), the analysis of the vast majority of protocols from the literature carries over between the two scheduler models without significant changes. This is the case, in particular, for clock protocols, such as our base oscillator protocol from~\cite{DBLP:journals/corr/DudekK17}, and all the asymptotic results on the performance of such clock protocols remain intact under the random-matching scheduler model.

By considering protocol $P$ to be a copy of clock $C^{(1)}$, we obtain a new clock $C^{(j+1)}$ whose rate is $r^{(j+1)} = \Theta(r^{(1)} r^{(j)})$, with $r^{(j+1)} \geq r^{(1)} r^{(j)}$. Recalling that $r^{(1)} = \alpha \ln n$, we obtain in particular for all $j$ constant: $r^{(j+1)} \geq \alpha \ln n r^{(j)}$ and $r^{(j)} = \Theta((\alpha \ln n)^j)$. This provides us with a hierarchy of clocks, each ticking at polylogarithmic rate, with clock $C^{(j)}$ performing at least $\alpha \ln n - O(1)$ cycles during one cycle of clock $C^{(j+1)}$. We choose the hierarchy of clocks to have constant height, more precisely, $j \leq l_{\max}$, where we recall that $l_{\max}$ is the maximum loop depth of the executed sequential code. By applying union bounds, we can thus also ensure that all clocks in the hierarchy are operating correctly, w.h.p.

The details of the construction provide one more property of the clock hierarchy which simplifies subsequent analysis: while all clocks are operating correctly and at least one agent in the population is in state $C^{(j)}_0$ (or, more generally, in a state divisible by 4 for clock $C^{(j)}$), no update to the current copy of state of any of the clocks $C^{(j')}$, for $j' > j$, is being performed by any agent. We can thus compose the executed protocol with an auxiliary protocol, which at the beginning of a clock cycle for clock $C^{(j)}$, stores for each agent in variable $C^{*(j+1)}_i$ a local copy of the state variable $C^{*(j+1)}_i$:

\RULE{C^{(j)}_0 \wedge C^{(j+1)}_i}{C^{(j)}_0}{C^{(j)}_0 \wedge C^{*(j+1)}_i}{C^{(j)}_0}

\RULE{C^{(j)}_0 \wedge \neg C^{(j+1)}_i}{C^{(j)}_0}{C^{(j)}_0 \wedge \neg C^{*(j+1)}_i}{C^{(j)}_0}

\noindent for all $j>0$. Note that not all agents will have the same values of $C^{*(j+1)}_i$, since the simulated clock $C^{(j+1)}$ might have been ``frozen'' in the process of updating its state while it was being copied. In this case, at most two states $C^{*(j+1)}_i$, $C^{*(j+1)}_{(i+1) \bmod m}$ may be set for different agents of the population. We fix this using a standard consensus rule applied strictly later (i.e., after a full tick of clock $C_j$), defaulting to the larger of the value:

\RULE{C^{(j)}_2 \wedge C^{*(j+1)}_i}{C^{(j)}_2 \wedge C^{*(j+1)}_k}{C^{(j)}_2 \wedge C^{*(j+1)}_i}{C^{(j)}_2 \wedge C^{*(j+1)}_i \wedge \neg C^{*(j+1)}_k}, for $i > k$.

Thus, from state $C^{(j)}_4$ until the end of the current cycle for clock $C^{(j)}$, each agent now holds stored (and does not update) the same value $C^{*(j+1)}_i$. We remark that when $C^{*(j+1)}_i$ is set, then always either $C^{(j+1)}_i$ is set or $C^{(j+1)}_{(i+1) \bmod m}$ is set for each agent (i.e., the locally stored state of the higher-level clock is at most one state behind the real state of that clock, for any agent), since the cycle length of clock $C^{(j)}$ is (much) more than $2$ times shorter than that of clock $C^{(j+1)}$.

As a direct corollary, we make the following crucial observation.
\begin{proposition}
Let $\sigma = (\sigma_{l_{\max}}, \ldots, \sigma_1)$ be an arbitrary vector with $\sigma_j \in \{4,\ldots, m-2\}$, for all $j \in \{1, \ldots, l_{\max}\}$. Suppose for some agent $C^{(1)}_{\sigma_1} \wedge \bigwedge_{j>1} C^{*(j)}_{\sigma_j}$ is set. Then, $\bigwedge_{j>1} C^{*(j)}_{\sigma_j}$ is also set for all agents in the population.
\qed
\end{proposition}

We now relate the period $m$ of clock $C_1$ to the compiled code, by setting $m := 4 w_{\max} + 2$. In relation to the above Proposition, we define the \emph{time path} $\tau = (\tau_{l_{\max}}, \ldots, \tau_1)$ of an agent as the unique vector such that $\tau_j \in \{1, \ldots, w_{\max}\}$ for all $j$, and  $C^{(1)}_{4 \tau_1} \wedge \bigwedge_{j>1} C^{*(j)}_{4\tau_j} =: \Pi_{\tau}$ is set for this agent; if no such vector exists, we put $\tau = \perp$.

The sequence of time paths observed in the program will completely determine its execution. We first formalize the recursively defined properties of defined clock hierarchy into iterative form, expressed by the following claim (which follows from a standard de-recursivization of the analysis).

\begin{proposition}\label{pro:exec}
The time paths of all agents satisfy the following properties w.h.p. during any time interval $I$ of at most polynomial length during which all clocks are operating correctly:
\begin{itemize}
\item At any moment of time, all time paths of agents which are not equal to $\perp$ are identical and equal to the same value $\tau$.
\item We call a time step $t$ in which all agents hold the same time path $\tau_t \neq \perp$ (i.e., no agent holds a time path of $\perp$ during step $t$) a \emph{good step}, and denote by $I_G \subseteq I$ the subset of good steps in interval $I$. Then, the sequence of time paths observed over good steps, $(\tau_t)_{t \in I_G}$ is
    a valid output of \emph{some} execution of the non-deterministic program with the pseudocode given in Fig.~\ref{fig:loopstructure}.
\item In the time interval between two good steps, the time path held by any agent is first the time path from the preceding good step, then $\perp$, and then the time path from the succeeding good step. \qed
\end{itemize}

\end{proposition}
\begin{figure}
\begin{protocol}
GoodStepCount := 0
\textbf{for} $\tau_{l_{\max}} := 1$ \textbf{to} $+\infty$ \textbf{do}:
    \textbf{for} $\tau_{l_{\max}-1} := 1$ \textbf{to} $w_{\max} \cdot$ RandInt[$\gamma \ln n, \delta \ln n$] \textbf{do}:
        \ldots
            \textbf{for} $\tau_{2} := 1$ \textbf{to} $w_{\max} \cdot$ RandInt[$\gamma \ln n, \delta \ln n$]  \textbf{do}:
                \textbf{for} $\tau_{1} := 1$ \textbf{to} $w_{\max} \cdot$ RandInt[$\gamma \ln n, \delta \ln n$]  \textbf{do}:
                    \textbf{for} RulesetLoop$ := 1$ \textbf{to} RandInt[$\gamma n \ln n, \delta n \ln n$]  \textbf{do}:
                        \textbf{print} $(\tau_{l_{\max}} \bmod m, \tau_{l_{\max}-1} \bmod m, \ldots, \tau_{l_2} \bmod m, \tau_{l_1} \bmod m)$
                        GoodStepCount := GoodStepCount + 1
                        \textbf{if} GoodStepCount $\geq |I| / (\delta \ln n)^{l_{\max}}$:
                            \textbf{if} RandInt[0,1] = 1:
                                \textbf{stop}.
\end{protocol}
\setcounter{figure}{0}
\caption{The sequence of time paths observed in the population at time steps of interval $I$ when all time paths of all agents are identical is w.h.p. a valid output of \emph{some} execution of the non-deterministic program with the pseudocode given above. $\gamma$ and $\delta$, with $0 < \gamma < \delta$ are constants, such that $\gamma$ can be chosen to be arbitrarily large and $\delta$ depends on the choice of $\gamma$. The routine RandInt returns an integer from the given interval, sampling each integer with (arbitrary) positive probability.}\label{fig:loopstructure}
\end{figure}

We remark that in Fig.~\ref{fig:loopstructure}, it is guaranteed that termination will only occur after at least $\geq |I| / (\delta \ln n)^{l_{\max}}$. This is because the interval between adjacent good steps does not exceed the time the outermost clock $C^{(l_{\max})}$ takes to traverse $6$ states, which is less than the length of its clock cycle.

\subsection{Deploying the Clock Hierarchy: Rules of the Compiled Protocol}\label{sec:compiledrules}
We design the rules of the compiled program, to be composed with the protocol running the clock hierarchy, so that agents select the current ruleset to be executed based solely on their timepath. Note that there exists a one-to-one correspondence between timepaths (different from $\perp$, which can be seen as an idle step) and leaves of tree $T$ identified by paths taken from the root of tree $T$ to the leaf; we will assume both timepaths and leaves are indexed by vectors $\tau = (\tau_{l_{\max}}, \ldots, t_{1})$, with $\tau_j \in \{1,\ldots, w_{\max}\}$. By convention, level $l_{\max}$ represents the outermost loop level of the code, and level $1$ the innermost. We proceed to identify leaves inductively: having fixed $(\tau_1, \ldots, \tau_{j-1})$, a value of $\tau_j$ corresponds to picking the $\tau_j$-th child (in top-down order of instructions) in the $j$-th-outermost loop of the code reachable along the path $(\tau_1, \ldots, \tau_{j-1})$. We denote by $R_{\tau}$ the ruleset associated with leaf $\tau$ (formally, leaf $\tau$ will contain the instruction \ASYNC[c \ln n] $R_\tau$).

\paragraph{Compilation procedure.} The entire process of compilation of the precompiled code tree is given as follows: for each $\tau = (\tau_{l_{\max}}, \ldots, t_{1})$, with $\tau_j \in \{1,\ldots, w_{\max}\}$, for each rule \RULE{\Sigma_1}{\Sigma_2}{\Sigma_3}{\Sigma_4} $\in R_{\tau}$, we add to the final protocol the rule:
\RULE{\Pi_\tau \wedge \Sigma_1}{\Pi_\tau \wedge \Sigma_2}{\Sigma_3}{\Sigma_4} where we recall: $\Pi_\tau = C^{(1)}_{4 \tau_1} \wedge \bigwedge_{j>1} C^{*(j)}_{4\tau_j}$.

\paragraph{Note on execution.} The execution of the rules over time for any agent is governed by Proposition~\ref{pro:exec}. The execution of interactions executing rules associated with any time path is always completed before execution of rules associated with the next timepath in the ordering from Fig.~\ref{fig:loopstructure} starts. Moreover, each timepath appears in all agents for at least $\delta n \ln n$ successive good steps, hence each ruleset is guaranteed to be repeated for a sufficiently long time to satisfy the constraints of the loops under a uniform scheduler on the population (choosing $\delta \geq c$). In other words, during good steps (which last sufficiently long), the applied filter $\Pi_\tau$ for the current time path $\tau$ corresponds precisely to the execution of the original ruleset $R_\tau$. (Note that directly before and after the good steps for time path $\tau$, ruleset $R_\tau$ may be executed by some subset of the agents in the population; this is tolerated behavior as per the specification of the programming language.) Finally, the order in which time paths appear (cf.~Fig.~\ref{fig:loopstructure}) is precisely the order of leaves in a correct sequential execution of rules in the tree.

The number of rounds taken to perform a complete iteration of the outermost infinite ``for'' loop (i.e., to iterate through all possible time paths) is at most poly logarithmic in $n$, given as $O((\delta\log n)^{l_{\max} + 1})$. This value asymptotically determines the required time the clocks need to operate correctly before stopping to achieve intended operation of the program.

\section{Programming Always-Correct Protocols}

\subsection{Exact Leader Election Protocol}
\label{leexact}

This Section is devoted to the analysis of protocol \fun{\LeaderElectionExact}, which is version of protocol \fun{\LeaderElection} modified to achieve correct computation with certainty.
\begin{figure}
\begin{protocol}
\PROTOCOL{\LeaderElectionExact}
\VAR{\OUTPUT{\INIT{L}{\on}}, \INIT{R}{\on}, \INIT{F}{\on}}

\THREAD{\Main}{L, \READONLY{R, F}}
    \LOCAL{\INIT{D}{\off}}
    \REPEAT
        \IFNONEMPTY{L}
            \GETS{D}{L \wedge F}
            \IFNONEMPTY{D}
                \GETS{L}{L \wedge D}
        \ELSE
            \GETS{L}{R}

\THREAD{\FilteredCoin}{F}
    \LOCAL{\INIT{I}{\on}, \INIT{S}{\on}}
    \ASYNC
        \RULE{I}{I}{\neg I \wedge S}{\neg I \wedge \neg S}
        \RULE{I}{\neg I}{\neg I}{\neg I}
        \RULE{S}{\neg S}{S \wedge F}{S \wedge F}
        \RULE{\neg S}{S}{\neg S \wedge F}{\neg S \wedge F}
        \RULE{F}{.}{\neg F}{.}

\THREAD{\ReduceSets}{R,L}
    \ASYNC
        \RULE{R}{R \wedge \neg L}{R}{\neg R \wedge \neg L}
        \RULE{R \wedge L}{R \wedge L}{R \wedge L}{\neg R \wedge \neg L}
\end{protocol}
\end{figure}

\begin{theorem}
\fun{\LeaderElectionExact} eventually elects a leader. For a protocol compilation following Theorem~\ref{th:guarantees}(ii)(b), a correct result with certainty is reached within $O(\poly(n))$ steps, w.h.p., starting from any reachable state of the protocol.
\end{theorem}
\begin{proof}
Eventually (w.h.p. after at most polynomial time), $\set F = \emptyset$ by thread \fun{\FilteredCoin}, and set $\set F$ never changes again.
Eventually (w.h.p. after at most polynomial time), $|\set R| = 1$ by thread \fun{\ReduceSets}, and set $\set R$ never changes again (e.g., $\set R = \{a\}$, for some agent $a$ in the population).
At the end of the next good iteration after the above, in thread \fun{\Main}, we have $\set D = \emptyset$ and set $\set D$ never changes again.
At the end of the next good iteration after the above, we have $\set L = \set R = \{a\}$.
Any successive update to $L$ is of the form \GETS{L}{R}. It follows that $\set L = \{a\}$ subsequently always holds.
\end{proof}

\begin{theorem}
\fun{\LeaderElectionExact} elects a leader in $O(\log^2 n)$ parallel rounds w.h.p., counting from the end of the initialization phase from  Theorem~\ref{th:guarantees}(ii)(b).
\end{theorem}
\begin{proof}
We remark that, following the above analysis, the correct result with certainty is reached in expected polynomial time. Thus, to asymptotically bound the expected time until the correct result is reached, it suffices to show that the Main thread reaches a correct result w.h.p. after at most a poly-logarithmic number of good iterations of its outer loop. The bound on the expected time until a correct result is reached will then follow from the promise of the framework.

First, we observe that there is always $|\mathcal{R}| \ge 1$, by the rules of thread \fun{\ReduceSets}. We now analyze thread \fun{\FilteredCoin}. Set $\mathcal{I}$ starts at size $n$ and eventually becomes empty set (with high probability in $O(\log n)$ rounds). Consider first $n/2$ nodes eliminated from $\mathcal{I}$. While $|\mathcal{I}| \ge n/2$, rule from line 17 is at least as probable to trigger as rule from line 18. Thus, with very high probability there are at least $n/8$ triggers of line 17, adding $n/8$ elements to $\mathcal{S}$. On the other hand, line 17 was triggered at most $n/2$ times. By the lines 19 and 20, $\mathcal{S}$ with the same probability increases by one element and decreases by one element, from which we have that for the next $\Omega(n)$ rounds there is $\frac{n}{16} \le |\mathcal{S}| \le \frac{15}{16}n$. Denote by $f_i$ the size of set $F$ in the step $i$. There is $\mathbb{E}[ f_{i+1} - f_{i} | f_{i}] \le \frac25 (1-\frac{f_i}{n}) \cdot 2 - \frac15 \frac{f_i}{n} = \frac{4}{5} - \frac{f_i}{n}$, where the bound follows from upperbounding the probability of rules 19 and 20 successfully adding node to $F$. Similarly, $\mathbb{E}[ f_{i+1} - f_{i} | f_{i} ] \ge \frac{2}{5} \cdot \frac{15}{16} \cdot \frac{1}{16} \cdot 2 \cdot (1-\frac{f_i}{n}) - \frac15\frac{f_i}{n} = \frac{3}{64} - \frac{79}{320}\frac{f_i}{n}$. Thus the fixed point is between $\frac{15}{79}n$ and $\frac{4}{5}n$. Thus, by Azuma's inequality, with very high probability $\frac{15}{158}n \le |F| \le \frac{9}{10}n$, starting from round $O(\log n)$, for at least $\Omega(n)$ rounds. Moreover, for any agent $x$, by the above analysis, there is in each interaction a $\Theta(1/n)$ probability of $x \in \mathcal{F}$ or $x \not\in \mathcal{F}$ being overwritten by an application of one of the rules from lines 19--21, which translates to high probability over any $\Omega(\log n)$ rounds. Thus if we consider $\mathcal{F}_0, \mathcal{F}_1, \mathcal{F}_2, \ldots$ being set $\mathcal{F}$ in a consecutive round snapshots, then for $i,j = \Omega(\log n)$ such that $i - j \ge t = c \log n$, the sets $\mathcal{F}_i$ and $\mathcal{F}_j$ are independent, conditioned on an event which holds with high probability $1 - 1/n^{\Omega(c)}$.

We consider now the first good iteration. If $\mathcal{L}$ is empty, it will be set to $\mathcal{R}$ in the next iteration, and $\mathcal{R}$ is nonempty. Thus w.l.o.g., assume $\mathcal{L}$ is non-empty. Denote by $\ell_i$ the size of $\mathcal{L}$ in $i$-th iteration after this one. Denote by $\frac{1}{11} \le  p_i \ge \frac{10}{11}$ the probability of success of the synthetic coin $\mathcal{F}$ in $i$-th iteration.  There is $\mathbb{E}[\ell_{i+1} | \ell_i] \le \ell_i \cdot p_i + (1-p_i)^{-\ell_i} \cdot \ell_i$, so for $\ell_i \ge 2$ there is $\mathbb{E}[\ell_{i+1} | \ell_i] \le \frac{111}{121} \ell_i$. It follows that for some $t = c \log n$ for large enough constant $c$, $\ell_t = 1$ with high probability.
We have that $O(c \log n)$ uncorrelated sets $\mathcal{F}$ are enough to narrow down $\mathcal{L}$ to size 1. Thus, in $O(\log^2 n)$ rounds this happens with high probability, as every $\mathcal{F}_{i \cdot c \log n}$, $i \ge t$ for $t$ large enough constant is uncorrelated with high probability, and all the additional $\mathcal{F}_{j}$ for $i c \log n < j < (i+1) c \log n$ only possibly speed up the process. To complete the proof, we observe that line 26 in thread \fun{\ReduceSets} only speeds up the process of narrowing down $\mathcal{L}$, while never reducing $\mathcal{L}$ to empty set. Additionally, if at some point $|\mathcal{L}| = 1$, then in $O(n)$ rounds in expectation, $|\mathcal{R}| = 1$ by the line 25, and after that $\mathcal{R}$ remains constant.
\end{proof}

\subsection{Exact Majority Protocol}
\label{majexact}

This Section is devoted to the analysis of protocol \fun{\MajorityExact}, provided in the form of pseudocode, being a version of protocol \fun{\Majority} adapted to achieve correctness with certainty.

\begin{figure}
\begin{protocol}
\PROTOCOL{\MajorityExact}
\VAR{\OUTPUT{Y_A}, \INPUT{A, B}}

\THREAD{\Main}{Y_A, {A, B}}
    \LOCAL{\INIT{A^*}{\off}, \INIT{B^*}{\off}, \INIT{K}{\off}}
    \REPEAT
        \ASYNC[c \ln n]
            \RULE{A}{B}{\neg A}{\neg B}
        \GETS{A^*}{A}
        \GETS{B^*}{B}
        \REPEAT[c \ln n]
            \ASYNC[c \ln n]
                \RULE{A^*}{B^*}{\neg A^*}{\neg B^*}
            \GETS{K}{\off}
            \ASYNC[c \ln n]
                \RULE{A^* \wedge \neg K}{\neg A^* \wedge \neg B^*}{A^* \wedge K}{A^* \wedge K}
                \RULE{B^* \wedge \neg K}{\neg A^* \wedge \neg B^*}{B^* \wedge K}{B^* \wedge K}
        \IFNONEMPTY{A^*}
            \GETS{Y_A}{\on}
        \IFNONEMPTY{B^*}
            \GETS{Y_A}{\off}
\end{protocol}

\end{figure}

\begin{theorem}
Protocol \fun{\MajorityExact} eventually correctly computes majority.
The result is obtained in $O(\log^3 n)$ parallel rounds w.h.p., counting from the end of the initialization phase from  Theorem~\ref{th:guarantees}(ii)(b).
\end{theorem}
\begin{proof}
Suppose w.l.o.g.\ that initially $|\set A| < |\set B|$.
Eventually (w.h.p. after polynomial time), $|\set A| = 0$ 
and set $\set A$ never changes again. At the end of the next good iteration after the above 
we have $\set A^* = \emptyset$ and set $\set A^*$ never changes again. At the end of the next good iteration after the above, we have $\set Y_A = \emptyset$ and $\set Z_{(18)} = \emptyset$. Set $\set Z_{(3)}$ never changes again. Any successive valuation of $Z_{(18)}$ in the compilation of line ``$\IFNONEMPTY{A^*}$'' always returns false, hence the branch of the program containing the line \GETS{Y_A}{\on} is never entered. It follows that any successive update to $Y_A$ is of the form \GETS{Y_A}{\off}, ans so $\set Y_A = \emptyset$ subsequently always holds. The case of $|\set A| > |\set B|$ is analogously analyzed, reaching $\set Y_A$ equal to the entire population.

For a protocol compilation following Theorem~\ref{th:guarantees}(ii)(b), a correct result with certainty is reached within $O(\poly(n))$ steps, w.h.p., starting from any reachable state of the protocol.

The expected time it takes for the protocol to compute majority follows from the analysis identical to the one done in Theorem~\ref{th:majority_runtime}.
\end{proof}

\subsection{Computing Arbitrary Semi-linear Predicates}

\label{general-semi-linear}

Generalizing the approach used in the exact solution to majority, we can also solve arbitrary semi-linear predicates. For a fixed semi-linear predicate\footnote{For simplicity and w.l.o.g., we consider binary-valued predicates, only.} $\Pi$ on a finite set of input states $A_1, A_2, \ldots$, we achieve this by combining two other population protocols for computing $\Pi$ as black boxes. The first of these protocols, given in~\cite{DBLP:journals/dc/AngluinAE08a}, has the property that for a population with a distinguished leader state and a given unique leader ($|\mathcal{L}=1|$), after $O(\log^5 n)$ rounds of computation, it writes the value of $\Pi(\mathcal{A}_1,\mathcal{A}_2,\ldots)$ to the output of all agents, w.h.p. of correctness. We call it the \emph{fast} blackbox. Moreover, we use the generic technique for exactly and stably computing in expected polynomial time the value of $\Pi(\mathcal{A}_1,\mathcal{A}_2,\ldots)$ to the output of all agents~\cite{DBLP:journals/dc/AngluinADFP06}. This is called the \emph{slow} blackbox. The slow blackbox uses output states $(P_D^1, P_D^0)$, at most one of which is set for each agent, and stable computation means that if state $P_D^i$ is set for all agents, $i\in \{0,1\}$, then $i$ represents the true value of $\Pi(\mathcal{A}_1,\mathcal{A}_2,\ldots)$ .

The fast blackbox and slow blackbox are put together in separate threads with the threads of protocol \fun{\LeaderElectionExact}, to achieve protocol \fun{\SemilinearExact}. Note that the execution of the ruleset over $\geq (c \ln n)^5$ rounds should be read as $5$ nested loops over $\geq (c \ln n)$ rounds.

\begin{figure}[!!!ht]
\begin{protocol}
\PROTOCOL{\SemilinearExact} \{for predicate $\Pi$\}
\VAR{\OUTPUT{\INIT{P}{\on}}, \INPUT{A_1, A_2,\ldots}, \INIT{L}{\on}, \INIT{P_D}{\on}}

\{Import all threads of protocol \LeaderElection on variable $L$\}

\THREAD{SemiLinear}{P,\READONLY{L,P_D^0,P_D^1}}
\LOCAL{P^*}
    \REPEAT
        \{Reset all state variables of fast blackbox to initial settings using \GETS \}
        \ASYNC[(c \ln n)^5]
            $\triangleright$ \{Fast blackbox: w.h.p. compute $\Pi(\mathcal{A}_1,\mathcal{A}_2,\ldots)$ as $P^*$, using $\cal L$ as leader(s)\}
        \IFNONEMPTY{P^*}
	        \IFNONEMPTY{\neg P_D^0}        
		        \IFNONEMPTY{\neg P}        
    	            \GETS{P}{\on}
        \IFNONEMPTY{\neg P^*}
	        \IFNONEMPTY{\neg P_D^1}        
		        \IFNONEMPTY{P}        
    	            \GETS{P}{\off}
\THREAD{SemiLinearSlow}{P_D^0,P_D^1}
	\ASYNC
		$\triangleright$ \{Slow blackbox: deterministically compute $\Pi(\mathcal{A}_1,\mathcal{A}_2,\ldots)$ as $(P_D^0,P_D^1)$\}

\end{protocol}
\end{figure}

\begin{theorem}
Protocol \fun{\SemilinearExact} eventually correctly computes $\Pi(\mathcal{A}_1,\mathcal{A}_2,\ldots)$.
The result is obtained in $O(\log^5 n)$ parallel rounds w.h.p., counting from the end of the initialization phase from Theorem~\ref{th:guarantees}(ii)(b).
\end{theorem}
\begin{proof}
Eventually (w.h.p. after polynomial time), we have $P_D^i$ and $\neg P_D^{1-i}$ correctly and permanently set for all agents, where $i\in \{0,1\}$ is given as  $i = \Pi(\mathcal{A}_1,\mathcal{A}_2,\ldots)$, by the properties of the slow blackbox. By the analysis of protocol \fun{\LeaderElection}, eventually (w.h.p. after polynomial time), we have that $|\mathcal{L}| =1$. At the end of the next good iteration after both of the above conditions are reached, in thread \fun{SemiLinear}, we either have that (1) $P$ and $\neg P_D^0$ are set for all agents, or (2) $\neg P$ and $P_D^0$ are set for all agents (in which two cases we have converged to $P = \Pi(\mathcal{A}_1,\mathcal{A}_2,\ldots)$ and no further update to $P$ will be performed), or (3) we have $P\neq \Pi(\mathcal{A}_1,\mathcal{A}_2,\ldots)$ for some agent. In the last case, by the properties of the fast black box, at the end of the good iteration, we have $P^*=\Pi(\mathcal{A}_1,\mathcal{A}_2,\ldots)$ for all agents, w.h.p. Once this event holds, we will update $P$ accordingly, and no further updates to $P$ will ever be performed by the protocol.

For a protocol compilation following Theorem~\ref{th:guarantees}(ii)(b), a correct result with certainty is reached within $O(\poly(n))$ rounds, w.h.p., starting from any reachable state of the protocol.

Finally, we note that w.h.p., the correct value of $P = \Pi(\mathcal{A}_1,\mathcal{A}_2,\ldots)$ will be set within $O(\log^5 n)$ rounds from the time when the leader election protocol has converged, and will never change again, w.h.p., in the polynomial number of steps until the convergence of the slow blackbox. While the start of the blackbox is not perfectly synchronized and agents start times differ w.h.p. by $O(\log n)$, a careful inspection of the \cite{DBLP:journals/dc/AngluinAE08a} protocol reveals that this is not an issue, slowing down the computation by an additional $O(\log n)$ time. After the convergence of the slow blackbox, it will never change again, as analyzed above. Overall, it follows that the protocol converges within $O(\log^5 n)$ rounds from the end of the initialization phase.
\end{proof}
Since the length of the initialization phase of the protocol is $O(n^\eps)$ for the compilation scheme of Theorem~\ref{th:guarantees}(ii)(b), we obtain that the \fun{\SemilinearExact} protocol provides an exact result in $O(n^\eps)$ rounds in expectation. When using the compilation scheme of Theorem~\ref{th:guarantees}(ii)(a), we obtain a convergence time of $O(\log^5 n)$ rounds, but only w.h.p. of correctness of result.

\bibliographystyle{alpha}
\bibliography{population}

\end{document}